\documentclass[preprint,12pt]{article}

\usepackage{color}

\usepackage{amsmath}
\usepackage[all]{xy}

\usepackage{authblk}

\usepackage{enumerate}

\DeclareMathOperator{\sign}{sign}
\renewcommand{\phi}{\varphi}

\newcommand{\R}{\mathbf{R}}

\usepackage{graphicx}

\usepackage{amssymb}
\usepackage{amsthm}

\theoremstyle{plain}
\newtheorem{thm}{Theorem}[section]
\newtheorem{prop}[thm]{Proposition}
\newtheorem{lemma}[thm]{Lemma}

\begin{document}



\title{Exact Analysis of Intrinsic Qualitative Features of Phosphorelays using Mathematical Models}


\author[1,2]{Michael Knudsen}
\author[1]{Elisenda Feliu}
\author[1,*]{Carsten Wiuf}

\affil[1]{Bioinformatics Research Centre, Aarhus University, DK-8000 Aarhus C, Denmark}
\affil[2]{Centre for Membrane Pumps in Cells and Disease -- PUMPKIN, Aarhus University, DK-8000 Aarhus C, Denmark}
\affil[*]{\small{To whom correspondence should be addressed. Email: wiuf@birc.au.dk}}


\maketitle

\begin{abstract}
Phosphorelays are a class of signaling mechanisms used by cells to respond to changes in their environment. Phosphorelays (of which two-component systems constitute a special case) are particularly abundant in prokaryotes and have been shown to be involved in many fundamental processes such as stress response, osmotic regulation, virulence, and chemotaxis. We develop a general model of phosphorelays extending existing models of phosphorelays and two-component systems. We analyze the model analytically under the assumption of mass-action kinetics and prove that a phosphorelay has a unique stable steady-state. Furthermore, we derive explicit functions relating stimulus to the response in any layer of a phosphorelay and show that a limited degree of ultrasensitivity (the ability to respond to changes in stimulus in a switch-like manner) in the bottom layer of a phosphorelay is an intrinsic feature which does not depend on any reaction rates or substrate amounts. On the other hand, we show how adjusting reaction rates and substrate amounts may lead to higher degrees of ultrasensitivity in intermediate layers. The explicit formulas also enable us to prove how the response changes with alterations in stimulus, kinetic parameters, and substrate amounts. Aside from providing biological insight, the formulas may also be used to avoid time-consuming simulations in numerical analyses and simulations.
\end{abstract}






\section{Introduction}\label{Seq:Introduction}

Throughout the course of evolution, living organisms have developed a variety of different cellular mechanisms capable of responding to external stimulus, and post-translational modification of proteins is common to many of these mechanisms. In particular, modification by phosphorylation is widespread, and it is estimated that about 30\% of all proteins undergo modification by phosphorylation \cite{Cohen:2000tma}.

One particular type of phosphorylation mechanism is the so-called \emph{phosphorelay} in which a phosphate group is transferred via a series of proteins through binding \cite{Appleby:1996wm,Perraud:1999wk,Stock:2000dh,West:2001tx}. Phosphorelays are particularly abundant in prokaryotes, but systems have also been identified in eukaryotes. Common to all phosphorelays are two proteins, a histidine kinase (HK) and a response regulator (RR). Upon sensing external stimulus, a histidine residue on the HK autophosphorylates using ATP, and the phosphate group is transferred to an aspartate residue on the RR, either directly or through a series of intermediate steps. When phosphorylated, an output domain of the RR is capable of adjusting the cellular response.

Four examples of phosphorelays are shown in Figure~\ref{Fig:FourRelays}. The EnvZ/OmpR system in \emph{E.~coli} is involved in osmoregulation of porin genes \cite{Stock:2000dh,Russo:1991te}. Since it comprises only two components, the HK and the RR, it is also referred to as a \emph{two-component system} (TCS). A slightly more complicated TCS example is the BvgS/BvgA system in \textit{B.~pertussis}, used by the bacterium to activate virulence genes \cite{Stock:2000dh,Uhl:1996tj,Cotter:2003cj}, where the HK contains three phosphorylation sites. Some systems have one or more intermediate phosphotransfer modules, as is e.g.~the case for the osmoregulation pathway Sln1p/Ypd1p/Ssk1p in \textit{S.~cerevisiae} \cite{Stock:2000dh,Maeda:1994ba,Posas:1996ve} and the sporulation initiating pathway Spo0A/Spo0F/Spo0B/Spo0A in \textit{B.~subtilis} \cite{Stock:2000dh,Burbulys:1991wo,Perego:1996wv,Hoch:1993ch}.

\begin{figure}[!ht]
\begin{center}
\includegraphics[width=4in]{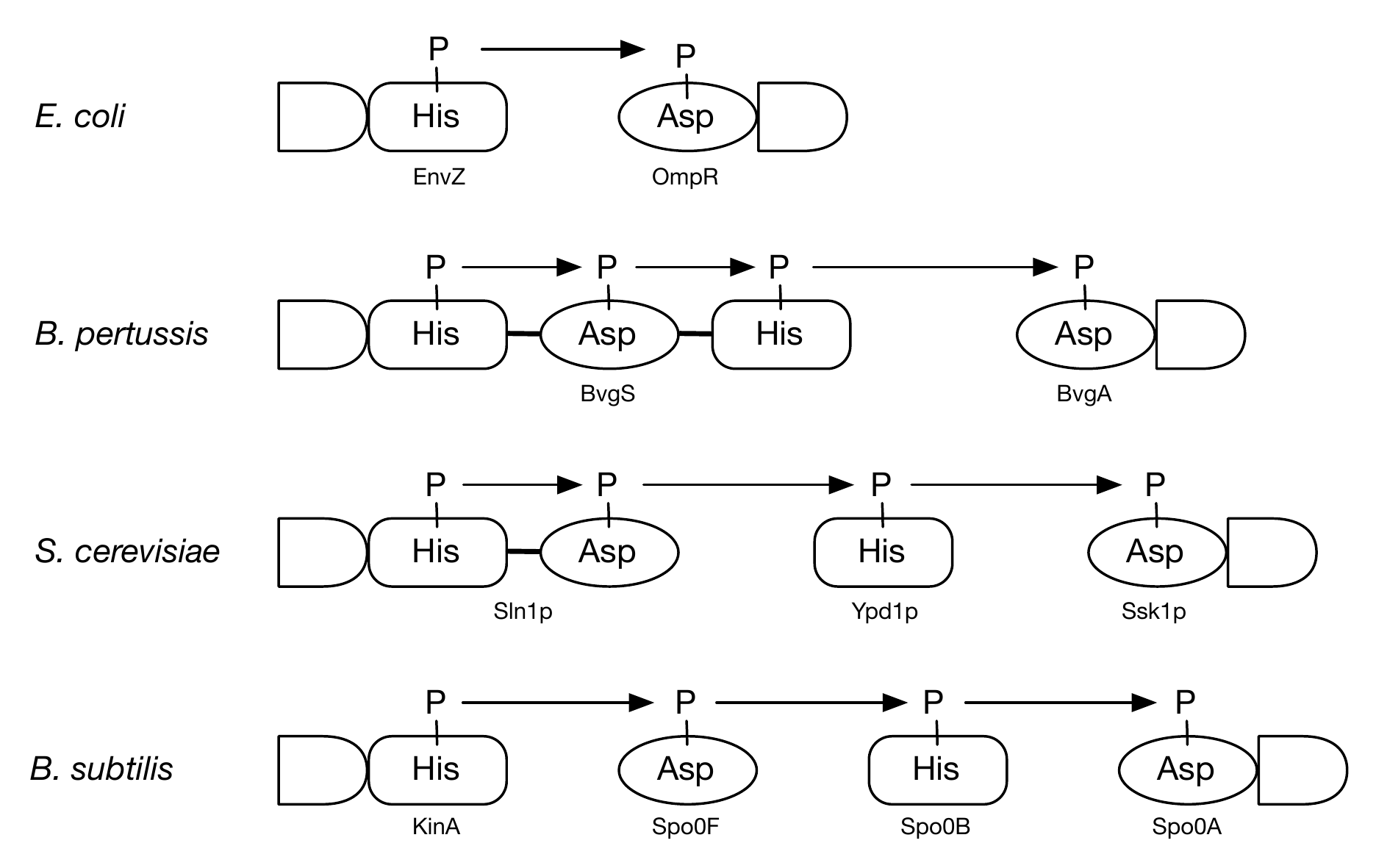}
\end{center}
\caption{
Examples of phosphorelays with different architectures. The EnvZ/OmpR and BvgS/BvgA are both examples of TCSs, but they vary in the number of phosphorylation sites on the HK. The systems in \emph{S. cerevisiae} and \emph{B. subtilis} both have a total of four phosphorylation sites, but they are distributed on three and four proteins, respectively.
}\label{Fig:FourRelays}
\end{figure}

The phosphorelays mentioned above are among the most well-described examples in the literature, but they only constitute a small fraction of the several hundreds of phosphorelays known \cite{Stock:2000dh,Chang:1998vx}, and studies of completed bacterial genomes have revealed the presence of many genes coding for HKs and RRs likely to be involved hitherto unknown phosphorelays \cite{West:2001tx,Zhang:2005kq}. For example, $62$ such genes have been identified in \emph{E. coli}, which amounts to more than $1\%$ of the entire genome \cite{West:2001tx}. Furthermore, the genes have been shown to be involved in a multitude of processes like stress response, osmotic regulation, virulence, and chemotaxis \cite{Mizuno:1997uca}, which illustrates the importance and ubiquity of phosphorelays.

Given the widespread occurrence of phosphorelays, it is only natural to ask what the benefits of such an elaborate signaling mechanism are. Among the phosphorelays known today, none have more than four phosphorylation sites in total \cite{Appleby:1996wm,Stock:2000dh}, however as illustrated in Figure~\ref{Fig:FourRelays}, the architectures may differ in the number of phosphorylation sites on each protein. One may thus speculate that whether the phosphorylation sites are located on one or more proteins influences the function of the phosphorelay, and that the benefits of a phosphorelay quickly saturate (or are balanced by drawbacks) with an increasing number of phosphorylation sites.

Mathematical modeling has been applied to study various types of biological networks, e.g~enzymatic reaction networks \cite{Gunawardena:2007gl,Gunawardena:2005jm,Manrai:2008kb,Salazar:2007bd,Salazar:2009tj,Thomson:2009dn,Kapuy:2009js,Wang:2007dc} and signaling cascades \cite{Goldbeter:1981tq,Goldbeter:1984vc,Huang:1996vm,Feliu:2011iu,Bluthgen:2006ec,Markevich:2004jo,Ventura:2008gc,Conradi:2008dw}, and has provided insight into steady-states, response to external stimulus, and robustness to changes in protein levels and kinetic parameters \cite{Gunawardena:2010dj, Shinar:2010dd, Shinar:2007gc, Li:2004eu, Batchelor:2003dy, Barkai:1997cd}. Precise measurements of concentrations and reaction rates are often difficult to obtain, and modeling can assist by determining whether e.g.~the number of steady states and the qualitative stimulus-response behavior is intrinsic to the network architecture and not dependent on the actual concentrations and reaction rates.

Here we develop a general model of phosphorelays of any size and architecture based on mass-action kinetics. The model extends existing models of phosphorelays \cite{Kim:2006hy,CsikaszNagy:2011dm}, and using an algebraic approach developed in \cite{Feliu:2011iu} we analytically analyze the model without resorting to numerical simulations. We prove the existence of a unique stable steady-state and show how it varies with changing model parameters. Furthermore, we obtain explicit expressions for stimulus-response curves. This allows us to derive an upper bound on the response coefficient in the bottom layer of any phosphorelay irrespectively of size and architecture, which is in agreement with what has been observed in both experiments and numerical models \cite{CsikaszNagy:2011dm,Fujita:2005fv}. Furthermore, we show that even for small phosphorelays (comprising only three phosphorylation sites), qualitatively very different response patterns are possible, and we derive explicit conditions on reaction rates and substrate concentrations describing each pattern. This contrasts what has previously been reported using simulation studies where saturation of phosphorylated sites at the bottom of the phosphorelay was suggested to cause a rise in response to sequentially propagate up through the phosphorelay \cite{CsikaszNagy:2011dm}.

Convergence and stability of the steady-state is proved using the theory of monotone dynamical systems  \cite{Angeli:2010ff,Angeli:2007ig}, which also provides the existence and uniqueness of the steady-state. However, our more direct algebraic approach to solving the steady-state equations is rewarded in that the calculations naturally extend to analytical results on the stimulus-response behavior. Combined with recent systematic approaches for reducing the complexity of the equations to be solved \cite{Feliu:2011uza,Thomson:2009iu}, we hope that similar direct, analytical calculations will become tractable for other chemical reaction networks too. 

\section{The Model}\label{Sec:TheModel}

We consider a general phosphorelay system consisting of $M\ge 2$ substrates $S^1,S^2,\ldots,S^M$, where the $m$th substrate $S^m$ has $N_m\ge 1$ phosphorylation sites (see Fig.~\ref{Fig:RelayCartoon}). We assume that substrates are never phosphorylated at more than one site at a time and denote by $S^m_n$ the $m$th substrate phosphorylated at its $n$th site with $n=0$ corresponding to the unphosphorylated state. We refer to the set of all phosphoforms $S^m_n$ with $0\le n\le N_m$ as the $m$th \emph{layer} of the phosphorelay and to $N_m$ as the \emph{length} of the $m$th layer.

We assume mass-action kinetics and that phosphate transfer within a substrate happens sequentially
\begin{align}\label{Reactions:IntraSubstrate}
\xymatrix{
S^m_1 \ar@<0.5ex>[r]^(0.50){a_{m,1}}&
S^m_2 \ar@<0.5ex>[r]^(0.50){a_{m,2}}\ar@<0.5ex>[l]^(0.50){b_{m,2}}&
\cdots \ar@<0.5ex>[r]^(0.45){a_{m,N_m-1}}\ar@<0.5ex>[l]^(0.50){b_{m,3}}&
S^m_{N_m} \ar@<0.5ex>[l]^(0.55){b_{m,N_m}}
}\qquad\text{for}\qquad 1\le m\le M,
\end{align}
with positive reaction constants $a_{m,n}$ and $b_{m,n}$, and we refer to these as \emph{intralayer} reaction rates. The transfer of phosphate groups between substrates in two different layers is modeled via the formation of intermediate complexes
\begin{align}\label{Reactions:InterSubstrate}
\xymatrix{
S^m_{N_m} + S^{m+1}_{0} \ar@<0.5ex>[r]^(0.65){u_m}&
X^m \ar@<0.0ex>[r]^(0.35){w_m}\ar@<0.5ex>[l]^(0.35){v_m}&
S^m_{0} + S^{m+1}_{1} 
}\qquad\text{for}\qquad 1\le m< M,
\end{align}
with positive reaction constants $u_m$, $v_m$, and $w_m$. That is, only when $S^m$ is phosphorylated at its last site can it transmit phosphate to the next layer.

Finally, we assume constant rates of phosphorylation (resp.~dephosphorylation) of $S^1_0$ and $S^M_{N_M}$, respectively, represented by two reactions
\begin{align}\label{Reactions:Special}
\xymatrix{
S^1_{0} \ar@<0.0ex>[r]^(0.5){c}& S^1_{1}
}\qquad\text{and}\qquad
\xymatrix{
S^M_{N_M} \ar@<0.0ex>[r]^(0.5){d}& S^M_{0}.
}
\end{align}

\begin{figure}[!ht]
\begin{center}
\includegraphics[width=4in]{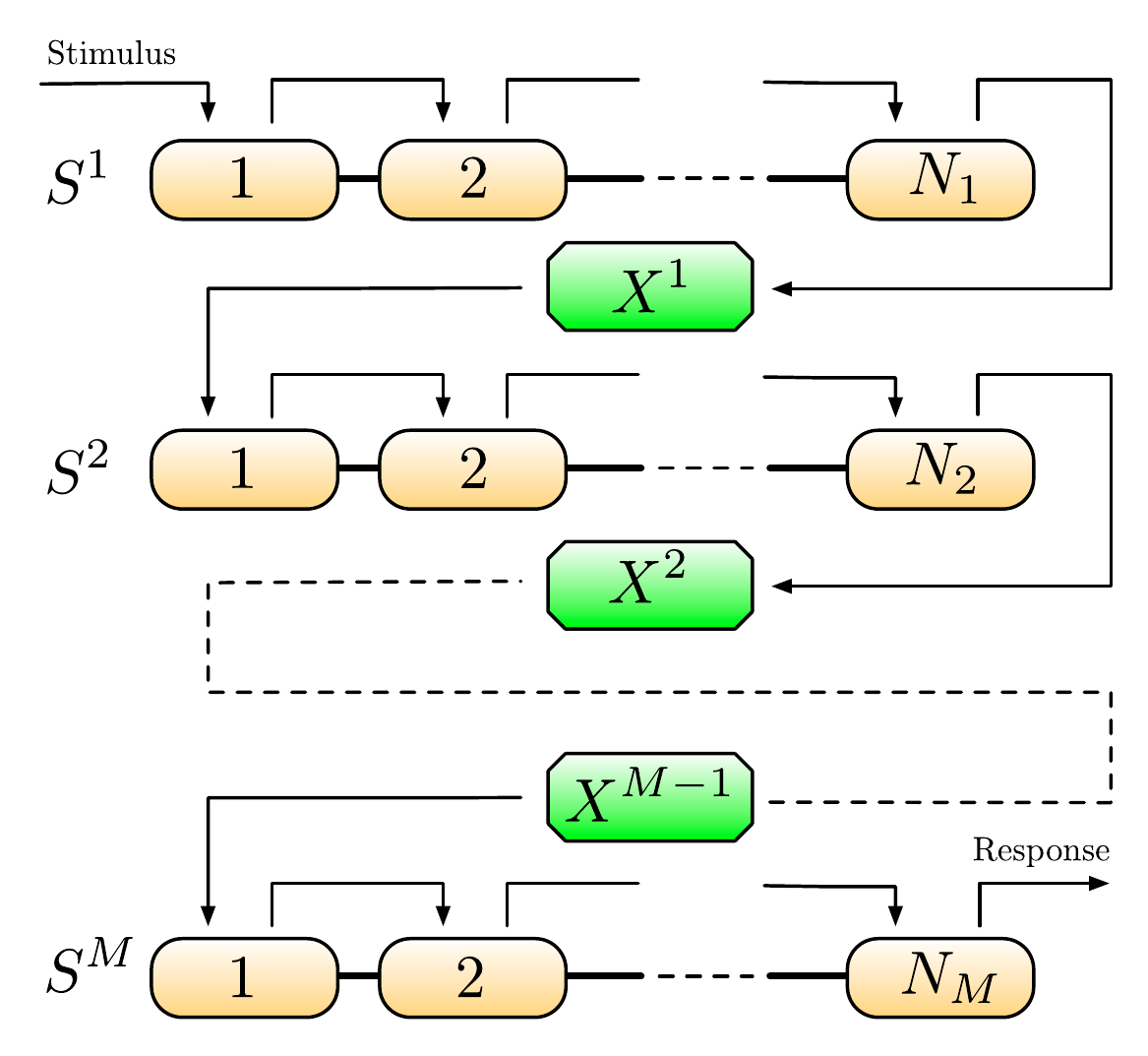}
\end{center}
\caption{
Schematic illustration of a general phosphorelay comprising $M$ layers. The number of phosphorylation sites in the $m$th layer is $N_m$, and the transfer of phosphate from one layer to the next is mediated via an intermediate complex.
}\label{Fig:RelayCartoon}
\end{figure}

The system is essentially linear with each new layer introducing a  new substrate. However, the mechanism of phosphorylation in the top layer is different from the phosphotransfer mechanism between the subsequent layers. Similarly, dephosphorylation of the bottom layer is different from the dephosphorylation mechanism in the other layers, which is also phosphotransfer. 

The \emph{stimulus} activating the relay is implicitly captured in the reaction constant $c$. Increasing $c$ corresponds to increasing the stimulus. When $c$ is very low, most of the substrate $S^1$ will remain unphosphorylated, whereas higher values of $c$ will push the substrate towards the phosphorylated phosphoforms. The \emph{final response} $S^M_{N_M}$ transmits its phosphate group to a receptor molecule, and this is modeled as a loss of the phosphate group without details about other molecules potentially involved in the process.

To avoid cumbersome notation, we denote by $S$ both the species $S$ and its concentration. It should always be clear from the context what is meant. Under the assumption of mass-action kinetics, the reactions \eqref{Reactions:IntraSubstrate}, \eqref{Reactions:InterSubstrate}, and \eqref{Reactions:Special} give rise to a set of differential equations,
\begin{align*}
\dot{S}^1_0 &= - c S^1_0 + w_1 X^1\\
\dot{S}^m_0 &= - u_{m-1} S^{m-1}_{N_{m-1}} S^m_0 + v_{m-1} X^{m-1} + w_m X^m&1< m< M\\
\dot{S}^M_0 &= - u_{M-1} S^{M-1}_{N_{M-1}} S^M_0 + v_{M-1} X^{M-1} + d S^M_{N_M}\\
\dot{S}^1_1 &= c S^1_0 - a_{1,1} S^1_1 + b_{1,2} S^1_2\\
\dot{S}^m_1 &= w_{m-1} X^{m-1} - a_{m,1} S^m_1 + b_{m,2} S^m_2&1< m\le M\\
\dot{S}^m_n &= a_{m,n-1} S^m_{n-1} + b_{m,n+1} S^m_{n+1} - (b_{m,n} + a_{m,n}) S^m_n&1\le m\le M&,\quad 1<n<N_m\\
\dot{S}^m_{N_m} &= - u_m S^m_{N_m} S^{m+1}_0 + v_m X^m + a_{m,N_m-1} S^m_{N_m-1} - b_{m,N_m} S^m_{N_m}&1\le m< M\\
\dot{S}^M_{N_M} &= a_{M,N_M-1} S^M_{N_M-1} - ( b_{M,N_M} + d ) S^M_{N_M}\\
\dot{X}^m &= u_m S^m_{N_m} S^{m+1}_0 - (v_m + w_m) X^m&1\le m< M
\end{align*}
By direct inspection it follows that
\begin{align}\label{Eq:PreConservationLaws}
\dot{S}^m_0+\dot{S}^m_1+\cdots+\dot{S}^m_{N_m}+\dot{X}^{m-1}+\dot{X}^m=0,
\end{align}
for all $1\le m\le M$, where we have defined $X^0=X^M=0$ in order to simplify notation, and hence the sum
\begin{align}\label{Eq:ConservationLaws}
S^m_{tot}&=S^m_0+S^m_1+\cdots+S^m_{N_m}+X^{m-1}+X^m,
\end{align}
is conserved for all $1\le m\le M$. This reflects the fact that $S^m$ either exists in one of its $N_m+1$ phosphoforms or is bound in one of the intermediate complexes $X^{m-1}$ or $X^m$. We will refer to $S^m_{tot}$ as the \emph{total amount} of the substrate $S^m$ and to \eqref{Eq:ConservationLaws} as the \emph{conservation law} for $S^m$.

In the following section we prove that for fixed reaction constants and total amounts of substrate, the phosphorelay has a unique steady-state, and we use the insight obtained in the proof to investigate the stimulus-response behavior of the system.

\section{Results}\label{Sec:Results}

\subsection{Steady-state equations}\label{SubSec:SteadyStateEquations}

The steady-state equations are the differential equations equated to zero along with the conservation laws for positive total amounts $S^m_{tot}$, and the steady-states are found by solving these for the variables (substrate phosphoforms and intermediate complexes). Hence there is a steady-state equation corresponding to each species as well as $M$ additional conservation laws. Since there are $M-1$ intermediate complexes, and each substrate $S^m$ exists in $N_m+1$ different phosphoforms, it follows that the system consists of $3M-1+\sum_{m=1}^MN_m$ equations in $2M-1+\sum_{m=1}^MN_m$ variables.

To obtain a simpler system of equations that more clearly elucidates the constraints imposed by the phosphorelay structure, we manipulate the steady-state equations to obtain a simpler, but equivalent, set of equations.

First note that according to \eqref{Eq:PreConservationLaws}, the equations $\dot{S}^m_1=0$ for $1\le m\le M$ hold if $\dot{S}^m_n=0$ and $\dot{X}^m=0$ hold for all $m$ and $n\neq 1$, and we may therefore leave them out. For all $1\le m<M$, the steady-state equation $\dot{X}^m=0$ is equivalent to the equation $\dot{X}^m+\dot{S}^{m+1}_0=0$, which in turn is equivalent to
\begin{align}\label{NewEq:SolutionXm}
X^m=\frac{d}{w_m}S^M_{N_M}\qquad\text{for}\qquad 1\le m<M.
\end{align}
Furthermore, by replacing all $\dot{S}^m_{N_m}=0$ by the equivalent $\dot{S}^m_{N_m}-\dot{S}^{m+1}_0=0$, it follows after inserting \eqref{NewEq:SolutionXm} that this is equivalent to
\begin{align}\label{Eq:SubResult1}
S^m_{N_m-1}=\frac{b_{m,N_m}S^m_{N_m}+d S^M_{N_M}}{a_{m,N_m-1}}\qquad \text{for}\qquad 1\le m\le M.
\end{align}
For $1\le n<N_m-1$, the steady-state equation $\dot{S}^m_{n+1}=0$ is equivalent to
\begin{align}\label{Eq:SubResult2}
S^m_n=\frac{(a_{m,n+1}+b_{m,n+1})S^m_{n+1}-b_{m,n+2}S^m_{n+2}}{a_{m,n}}\qquad \text{for}\qquad 1\le m\le M,
\end{align}
and using induction, we may combine \eqref{Eq:SubResult1} and \eqref{Eq:SubResult2} in one equivalent statement,
\begin{align}\label{NewEq:SolutionSmnABC}
S^m_{n}=\frac{B_{m,n}S^m_{N_m}+d\,C_{m,n}S^M_{N_M}}{A_{m,n}}\quad\text{for}\quad\begin{array}{l}1\le m\le M\\1\le n\le N_m,\end{array}
\end{align}
where the constants are defined by
\begin{align}\label{Eq:HandyConstants}
A_{m,n}=\prod_{i=n}^{N_m-1} a_{m,i},\:\:\:
B_{m,n}=\prod_{i=n+1}^{N_m} b_{m,i},\:\:\:
C_{m,n}=\sum_{i=n+1}^{N_m}\Big(A_{m,i}\prod_{j=n+1}^{i-1}b_{m,j}\Big)
\end{align}
for $1\le m\le M$ and $0\le n\le N_m$. In particular, these definitions imply that $A_{m,N_m}=B_{m,N_m}=1$, $C_{m,N_m}=0$, and $C_{m,N_m-1}=1.$ Apart from $C_{m,N_m}$, the constants are all positive and depend only on the intralayer reaction constants $a_{m,n}$ and $b_{m,n}$ in the $m$th layer.

Using \eqref{NewEq:SolutionXm} and \eqref{NewEq:SolutionSmnABC}, we see that the conservation law \eqref{Eq:ConservationLaws} is fulfilled if and only if
\begin{align}\label{NewEq:LambdaAndMu}
S^m_0=S^m_{tot}-\lambda_m S^M_{N_M}-\mu_m S^m_{N_m}\qquad\text{for}\qquad 1\le m\le M,
\end{align}
with constants given by
\begin{align}\label{Eq:LambdaMuConstants}
\lambda_m=d\Big( \frac{1}{w_{m-1}} + \frac{1}{w_m} + \sum_{n=1}^{N_m}\frac{C_{m,n}}{A_{m,n}} \Big)
\qquad\text{and}\qquad \mu_m=\sum_{n=1}^{N_m}\frac{B_{m,n}}{A_{m,n}}
\end{align}
(terms involving the undefined rates $w_0$ and $w_M$ are removed). The constants are all positive and depend only on reaction rates in the $m$th and $(m-1)$th layer. Finally, using \eqref{NewEq:SolutionXm} it follows that $\dot{S}^m_0=0$ is equivalent to
\begin{align*}
cS^1_0=dS^M_{N_M}\quad\text{and}\quad d\big(\frac{v_{m-1}}{w_{m-1}}+1\big)S^M_{N_M}=u_{m-1} S^{m-1}_{N_{m-1}} S^m_0\quad\text{for}\quad 1< m\le M.
\end{align*}
Summing up, the set of steady-state equations are replaced by an equivalent set of equations,
\[\begin{array}{rll}
(\textbf{SS1})&cS^1_0=d S^M_{N_M}\\[4pt]
(\textbf{SS2})&X^m=\frac{d}{w_m}S^M_{N_M}&1\le m<M\\[4pt]
(\textbf{SS3})&S^m_{n}=\frac{B_{m,n}S^m_{N_m}+d\,C_{m,n}S^M_{N_M}}{A_{m,n}}&1\le m\le M,\: 1\le n\le N_m\\[4pt]
(\textbf{SS4})&S^m_0=S^m_{tot}-\lambda_m S^M_{N_M}-\mu_m S^m_{N_m}&1\le m\le M\\[4pt]
(\textbf{SS5})&d\big(\frac{v_{m-1}}{w_{m-1}}+1\big)S^M_{N_M}=u_{m-1} S^{m-1}_{N_{m-1}} S^m_0&1< m\le M,
\end{array}\]
with constants defined in \eqref{Eq:HandyConstants} and \eqref{Eq:LambdaMuConstants}. Note that the reaction rate $c$ only appears in (\textbf{SS1}).

Throughout this paper we assume that all reaction constants and total amounts are fixed and  positive unless otherwise clearly stated. Any solution to the steady-state equations is  a steady-state, and the system could therefore possess multiple steady-states, some of which with negative concentrations. These are not biologically obtainable, so the focus is on steady-states in which all concentrations are non-negative (zero or positive). We call these \emph{biologically meaningful steady-states} (BMSSs).

\subsection{Existence of a unique stable BMSS}\label{Sec:ExistenceAndUniquenessOfBMSS}

In this section we prove the existence of a unique stable BMSS for a general phosphorelay. We do so by writing all steady-state concentrations as rational functions of the final response $S^M_{N_M}$ (recall that a \emph{rational function} in $x$ is a quotient $f(x)/g(x)$ of two polynomial functions in $x$) and then show that precisely one value of $S^M_{N_M}$ gives rise to a BMSS.

Starting with the $M$th layer, we work our way to the top layer by layer. The link between layers is obtained by relating the steady-state value of $S^m_{N_m}$ with that of $S^M_{N_M}$ through a rational function $S^m_{N_m}=\psi_m(S^M_{N_M})$. The singularities of $\psi_m$ for $1\le m\le M$ provide a necessary condition $S^M_{N_M}<\xi_m$ for non-negative concentrations in the layers $m,m+1,\ldots,M$, and we prove that $\xi_1<\xi_2<\cdots<\xi_{M-1}$, from which it follows that $S^M_{N_M}<\xi_1$ is necessary for all concentrations to be positive. We then write $c=\psi_0(S^M_{N_M})$ as an increasing rational function of $S^M_{N_M}$ and demonstrate how this leads to a stronger necessary condition $S^M_{N_M}<\xi_0$. Finally, we show that for any given value of $c$, the equation $c=\psi_0(S^M_{N_M})$ has a unique solution $S^M_{N_M}$ in $[0,\xi_0)$, which establishes the existence and uniqueness of a BMSS. In fact, it turns out that all steady-state concentrations are strictly positive.

Note that (\textbf{SS2})--(\textbf{SS4}) express $X^m$ and $S^m_n$ for $0\le n\le N_m$ as rational functions of $S^M_{N_M}$ and $S^m_{N_m}$ with coefficients depending on the intralayer reaction constants in the $m$th layer and the reaction constants $d$, $w_m$, and $w_{m-1}$ only. We now show how (\textbf{SS5}) yields the link to express all $S^m_{N_m}$ as rational functions of $S^M_{N_M}$.

We first show that at steady-state $S^m_0\neq 0$ for all $1\le m\le M$. If this is not the case, there is a largest $m$ for which $S^m_0=0$, and (\textbf{SS5}) then implies that $S^M_{N_M}=0$. For $m=M$, (\textbf{SS4}) implies that $S^M_{tot}=0$, which contradicts the assumption of positive total amounts. For $m<M$ we argue as follows: Since $m$ is the largest with the property $S^m_0=0$, we have $S^{m+1}_0\neq 0$, and combined with $S^M_{N_M}=0$, it follows from (\textbf{SS5}) that $S^m_{N_m}=0$. Now using (\textbf{SS4}) yields $S^m_{tot}=0$, which again contradicts the assumption of positive total amounts.

Since $S^m_0$ is non-zero at steady-state, we may isolate $S^{m-1}_{N_{m-1}}$ in (\textbf{SS5}) and use (\textbf{SS4}) to get
\begin{align}\label{NewEq:FormulaForSm-1Nm-1}
S^{m-1}_{N_{m-1}}=\frac{d\big(\frac{v_{m-1}}{w_{m-1}}+1\big)S^M_{N_M}}{u_{m-1}(S^m_{tot}-\lambda_m S^M_{N_M}-\mu_m S^m_{N_m})}
\qquad\text{for}\qquad 1< m\le M,
\end{align}
which shows that if we define $\psi_m$ recursively by $\psi_M=\operatorname{id}$, and
\begin{align}\label{NewEq:DefinitionOfPsi} 
\psi_{m-1}(y)=\frac{d\big(\frac{v_{m-1}}{w_{m-1}}+1\big)y}{u_{m-1}\big(S^m_{tot}-\lambda_m y-\mu_m \psi_m(y)\big)}
\qquad\text{for}\qquad 1<m\le M,
\end{align}
then $\psi_m(S^M_{N_M})=S^m_{N_m}$ at steady-state. The recursive definition implies that $\psi_m$ is a rational function. Furthermore, by isolating $c$ in (\textbf{SS1}) and inserting $S^1_0$ from (\textbf{SS4}), it follows using $S^1_{N_1}=\psi_1(S^M_{N_M})$ that $c=\psi_0(S^M_{N_M})$, where
\begin{align}\label{Eq:ThePsi0Function}
\psi_0(y)=\frac{d y}{S^1_{tot}-\lambda_1 y-\mu_1\psi_1(y)}
\end{align}
is also a rational function.

Writing $\psi_M(y)=p_M(y)/q_M(y)$ with $p_M(y)=y$ and $q_M(y)=1$, we may use \eqref{NewEq:DefinitionOfPsi} and \eqref{Eq:ThePsi0Function} to recursively write all $\psi_m(y)$ as quotients $p_m(y)/q_m(y)$ with $p_m(0)=0$, where both $p_m$ and $q_m$ are polynomials of degree $M-m$ for all $0\le m<M$.

\begin{prop}\label{NewProp:PsimSS}
The steady-state equations (\textbf{SS1})--(\textbf{SS5}) are satisfied, if and only if (\textbf{SS2})--(\textbf{SS4}) are satisfied along with $\psi_0(S^M_{N_M})=c$, and $\psi_m(S^M_{N_M})=S^m_{N_m}$ for all $1\le m\le M$.

Furthermore, for all $m<M$ the function $\psi_m$  has a minimal positive singularity $\xi_m$ satisfying
\begin{align*}
\xi_{M-1}>\xi_{M-2}>\cdots>\xi_1>\xi_0>0.
\end{align*}
Let $\xi_M=\xi_{M+1}=\infty$. Then $\psi_m$ is continuous, non-negative and strictly increasing on $[0,\xi_m)$, negative on $(\xi_m,\xi_{m+1})$, and it satisfies $\psi_m(0)=0$ and $\psi_m(y)\rightarrow\infty$ for $y\rightarrow\xi_m^-$ for all $0\le m\le M$.
\end{prop}

\begin{proof}
The first part of the proposition follows immediately since the equations $\psi_0(S^M_{N_M})=c$ and $\psi_m(S^M_{N_M})=S^m_{N_m}$ are just rearrangements of (\textbf{SS1}) and (\textbf{SS5}), respectively.

For the second part, the case $m=M$ is trivial since $\psi_M=\operatorname{id}$, and $(\xi_M,\xi_{M+1})$ is the empty set. Assume now that the claim is true for $m+1$ and consider the case $m$. By induction, $\psi_{m+1}$ is increasing on $[0,\xi_{m+1})$, so the denominator of $\psi_m$ is continuous and decreasing on $[0,\xi_{m+1})$, and it diverges towards $-\infty$ for $y\rightarrow\xi_{m+1}^-$. Therefore it has a unique zero $\xi_m<\xi_{m+1}$. Furthermore, the numerator of $\psi_m$ is continuous and increasing and equals $0$ for $y=0$, and therefore the entire fraction $\psi_{m}(y)$ is continuous, positive, and increasing on $[0,\xi_{m})$, negative on $(\xi_m,\xi_{m+1})$, $\psi_{m}(0)=0$, and $\psi_{m}(y)\rightarrow\infty$ for $y\rightarrow\xi_{m}^-$.
\end{proof}

\begin{thm}\label{Thm:SteadyStates}
For any set of fixed positive reaction constants and total amounts, the phosphorelay converges to a unique stable BMSS. In fact, the steady-state concentrations of all substrates and intermediate complexes are positive.
\end{thm}

\begin{proof}
It follows from Proposition \ref{NewProp:PsimSS} that $S^m_{N_m}\ge 0$ at steady-state, if and only if $S^M_{N_M}$ is in $[0,\xi_m)$, and since $\xi_1<\xi_2<\cdots<\xi_{M-1}$, it follows that $S^M_{N_M}<\xi_1$ is a necessary condition for a BMSS. According to (\textbf{SS2}), we have $X^m\ge 0$ for $1\le m\le M$ for any $S^M_{N_M}\ge 0$, and by inserting $S^m_{N_m}\ge 0$ into (\textbf{SS3}), it shows that also $S^m_n\ge 0$ for all $1\le m\le M$ and $1\le n<N_m$. Finally, because $\xi_m$ by definition is the smallest positive root of the right-hand side of (\textbf{SS4}) after substituting $S^m_{N_m}=\psi_m(S^M_{N_M})$, it follows that $S^m_0>0$ for all $1\le m\le M$. The argument also implies that all steady-state concentrations are positive if and only if $S^M_{N_M}>0$, and since $\psi_0(S^M_{N_M})=c>0$, this is always the case.

According to Proposition \ref{NewProp:PsimSS}, the function $\psi_0$ is continuous and increases from $0$ to $\infty$ on $[0,\xi_0)$ and is negative on $(\xi_0,\xi_1)$. It follows (see also Fig.~\ref{Fig:Psi0Function}) that precisely one value of $S^M_{N_M}$ in $[0,\xi_0)$ satisfies the condition $c=\psi_0(S^M_{N_M})$. This establishes the existence and uniqueness of a BMSS.

The convergence and stability part can be proved using methods from the theory of monotone dynamical systems (see Theorem 2 in \cite{Angeli:2010ff}), and a proof is included in \ref{Appendix:Stability}.
\end{proof}

\begin{figure}[htbp]
\begin{center}
\includegraphics[width=3.5in]{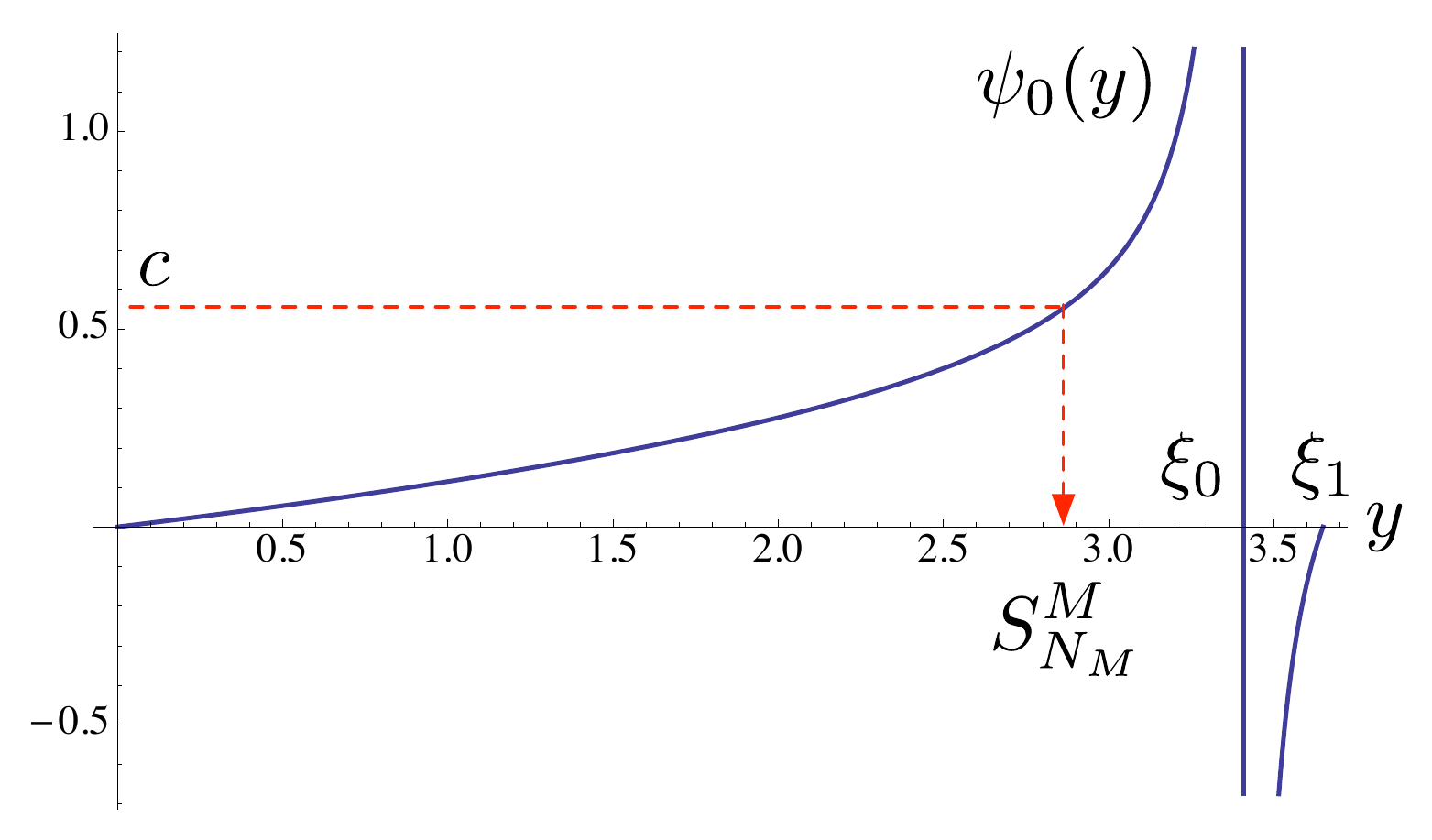}
\end{center}
\caption{
The graph of $\psi_0$ for a three-layer phosphorelay with all reaction constants equal to one and all total amounts equal to ten.
}\label{Fig:Psi0Function}
\end{figure}

The proof of convergence and stability in \ref{Appendix:Stability} also implies existence and uniqueness of a BMSS, but the proof is not constructive and does not yield the additional insight provided by the functions $\psi_m$. On the other hand, our approach does not address the convergence to and the stability of the steady-state, and the two methods thus complement each other.

\section{Stimulus-Response}\label{Sec:StimulusResponse}

In this section we demonstrate how the functions $\psi_m$ may be used to explicitly describe stimulus-response behavior.

\subsection{Maximal response}

Let all total amounts and all reaction constants but the stimulus $c$ be fixed. According to \eqref{Eq:ThePsi0Function}, the stimulus is an increasing continuous function of the response in $[0,\xi_0)$ and hence, \emph{vice versa}, the response is an increasing continuous function of the stimulus. Furthermore, $S^M_{N_M}\rightarrow\xi_0$ for $c\rightarrow\infty$, and therefore $\xi_0$ is the smallest upper limit on all possible responses. The limit is not attainable but can be thought of as the response in a fictitious system with infinite stimulus, and we will refer to it as the \emph{maximal response} of the phosphorelay.

As argued in Section \ref{Sec:ExistenceAndUniquenessOfBMSS}, the rational function $\psi_0$ is the ratio of two polynomials of degree $M$, and calculating the maximal response is thus equivalent to finding the smallest positive root in a polynomial of degree $M$.

More generally, in a phosphorelay with all total amounts and all reaction rates but $c$ fixed, we denote by $\rho_m$ the smallest upper limit of all possible steady-state values of $S^m_{N_m}$ and call it the \emph{maximal response} in the $m$th layer. We have just argued that $\rho_M=\xi_0$, and since according to Proposition \ref{NewProp:PsimSS} all $\psi_m$ are increasing functions on intervals containing $[0,\xi_0)$, we have
\begin{align}\label{Eq:AllMaximalResponses}
\rho_m=\psi_m(\rho_M)\qquad\text{for all}\qquad 1\le m<M.
\end{align}
Since $\psi_m$ is invertible (it is increasing and continuous), we have that $S^M_{N_M}=\psi_m^{-1}(S^m_{N_m})$, and by substituting this into \eqref{NewEq:DefinitionOfPsi} we obtain
\begin{align}\label{Eq:GeneralResponse}
c=(\psi_0\circ\psi_m^{-1})(S^m_{N_m})\qquad\text{for all}\qquad 1\le m\le M,
\end{align}
which is the stimulus expressed as a function of the response in the $m$th layer. Note that since \eqref{Eq:GeneralResponse} involves the inverse of a rational function, it is, in general, not itself a rational function.

The explicit stimulus-response relationship may be used to investigate how changes in one layer $m_0$ are reflected in the maximal responses in all layers of a phosphorelay. Suppose that $\lambda_{m_0}$ or $\mu_{m_0}$ is increased by changing reaction rates or by adding more phosphorylation sites to the substrate $S^{m_0}$ in an existing layer (see \eqref{Eq:LambdaMuConstants}). Then the maximal response decreases (resp.~increases) in layers below (resp.~above) $m_0$. Increasing the total amount $S^{m_0}_{tot}$ has the opposite effect. Then the maximal response increases (resp.~decreases) in layers downstream (resp.~upstream) from layer $m_0$. This is illustrated for $M=5$ in Figure \ref{Fig:StripLayers}A, and proofs of both claims are given in \ref{Prop:AppendixLayerEfftects}. The responses $S^m_{N_m}$ themselves exhibit the same behavior, and a proof of this is included in \ref{Prop:AppendixActualLayerEfftects}. Summing up, these results enable us to predict how all layers in the phosphorelay respond to changes in kinetic parameters and total amounts.

By removing the top layer from the system and adding a new stimulus reaction $S^2_0\stackrel{c}{\rightarrow}S^2_1$, we obtain a smaller phosphorelay with $M-1$ layers, and its maximal final response is the minimal, positive zero of $S^2_{tot}-\lambda_2 y-\mu_2\psi_2(y)$, which is exactly $\xi_1$ from Proposition \ref{NewProp:PsimSS}. In general, removing $m$ layers from the original system results in a smaller system with a larger maximal final response equal to $\xi_m$, and this is illustrated for $M=6$ in Figure \ref{Fig:StripLayers}B. With $m$ layers removed, the phosphorylation $S^{m+1}_0\stackrel{c}{\rightarrow}S^{m+1}_1$ is direct, whereas in the larger system, some of the $S^m_{N_M}$ taking part in the phosphorylation (phosphotransfer) is sequestrated in the intermediate complex $X^m$. 

\begin{figure}[htbp]
\begin{center}
\includegraphics[width=5.25in]{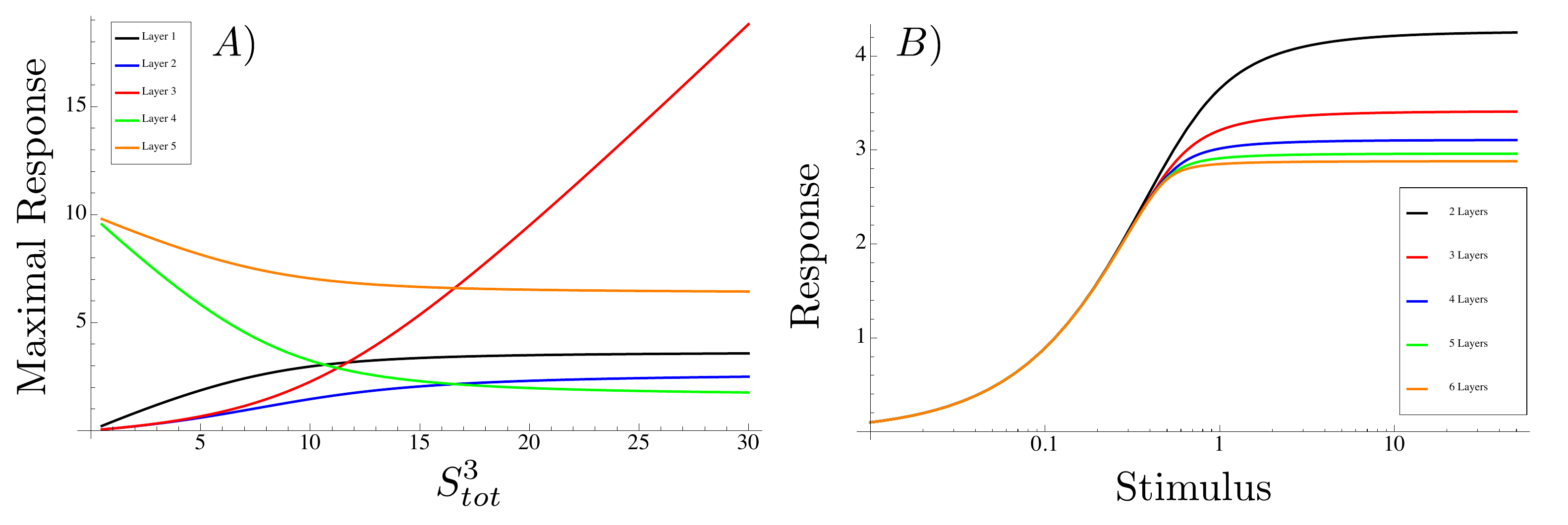}
\end{center}
\caption{
A) Maximal responses in all layers of a five-layer phosphorelay as functions of $S^3_{tot}$. When $S^3_{tot}$ increases, the maximal response increases in layers $1$, $2$, and $3$, but decreases in layers $4$ and $5$.  B) Stimulus-response curves for the bottom layer in phosphorelays with $6$, $5$, $4$, $3$, and $2$ layers, where smaller systems are obtained from larger by removal of upper layers. All reaction rates are set to one, and all total amounts are set to $10$.
}\label{Fig:StripLayers}
\end{figure}

\subsection{Ultrasensitive response}\label{SubSec:UltrasensitiveResponse}

In this section we use the functions $\psi_m$ to describe how steady-state concentrations respond to changes in stimulus. For any $0<\varepsilon<1$, we denote by $c_{m,\varepsilon}$ the amount of stimulus needed in order to obtain $\varepsilon$ times the maximal response in the $m$th layer. That is, using the notation introduced in \eqref{Eq:AllMaximalResponses}, we have
\begin{align*}
c_{m,\varepsilon}=(\psi_0\circ\psi_m^{-1})(\varepsilon\rho_m)=\psi_0(\psi_m^{-1}(\varepsilon\psi_m(\rho_M))).
\end{align*}
The \emph{normalized response} in the $m$th layer is the response $S^m_{N_m}$ divided by its maximal value $\rho_m$, and plotted as a function of $c$ we refer to it as the \emph{normalized stimulus-response curve} for the $m$th layer. The curve consists of the points $(c_{m,\varepsilon},\varepsilon)$ for $0\le\varepsilon\le 1$, that is $\varepsilon$ is the normalized response.

\begin{prop}\label{Prop:ResponseLastLayer}
For all $ m< M$ and $0<\varepsilon< 1$ we have $c_{M,\varepsilon}<c_{m,\varepsilon}$. That is, the normalized stimulus-response curve for the bottom layer is shifted to the left of the normalized stimulus-response curves for all other layers.
\end{prop}

The proof uses induction on $m$ and it is given in \ref{Prop:AppendixResponseLastLayer}. The result in Proposition \ref{Prop:ResponseLastLayer} cannot be extended to compare arbitrary layers, and we now demonstrate how it already fails for $M=3$. In fact, it turns out that, depending on the reaction rates and total amounts, we can have $c_{1,\varepsilon}>c_{2,\varepsilon}$ or $c_{1,\varepsilon}<c_{2,\varepsilon}$ for all $0<\varepsilon<1$, and in some cases the normalized stimulus-response curves for layers one and two intersect as illustrated in Fig~\ref{Fig:ThreePlotsExample}. Since $\psi_1$ and $\psi_0^{-1}$ are increasing functions, comparing $c_{1,\varepsilon}$ and $c_{2,\varepsilon}$ is equivalent to comparing
\begin{align}\label{Eq:ComparisonM3}
(\psi_1\circ\psi_0^{-1})(c_{1,\varepsilon})=
\varepsilon \psi_1(\rho_3)
\:\:\:\text{and}\:\:\:
(\psi_1\circ\psi_0^{-1})(c_{2,\varepsilon})=
\psi_1(\psi_2^{-1}(\varepsilon \psi_2(\rho_3))),
\end{align}
so it follows that comparing $c_{1,\varepsilon}$ and $c_{2,\varepsilon}$ is equivalent to determining the sign of $\Delta(\varepsilon)=\varepsilon \psi_1(\rho_3)-\psi_1(\psi_2^{-1}(\varepsilon \psi_2(\rho_3)))$. The expressions in \eqref{Eq:ComparisonM3} are easier to work with than the original ones, since we may calculate them explicitly. The maximal response $\rho_3$ is a root of a quadratic polynomial, and $\psi_2^{-1}$ may be obtained directly, since $\psi_2$ is the ratio of two first degree polynomials.

\begin{figure}[htbp]
\begin{center}
\includegraphics[width=5.4in]{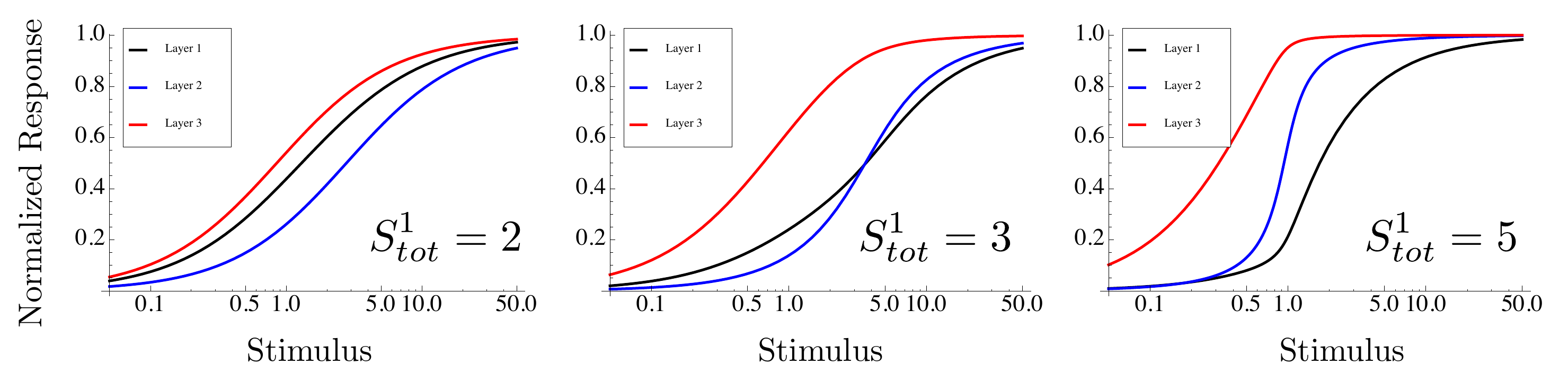}
\end{center}
\caption{
Normalized stimulus-response curves for a three-layer phosphorelay with one phosphorylation site at each layer. Here $S^2_{tot}=20$, $S^3_{tot}=5$, and all reaction rates are set to one. Then $\varepsilon^*=35/\rho_3-15$, where $\rho_3$ depends on $S^1_{tot}$, and by varying $S^1_{tot}$ we obtain three qualitatively different behaviors of the normalized response curves. The stimulus is on logarithmic scale.
}\label{Fig:ThreePlotsExample}
\end{figure}

By manipulating $\Delta(\varepsilon)$ (see \ref{Appendix:Example} for details), we see that its sign is determined by the roots of a polynomial of degree three in $\varepsilon$. We find that both $0$ and $1$ are roots, and the third root is
\begin{align*}
\varepsilon^*=\frac{w_2u_2(S^2_{tot}(\lambda_3+\mu_3)-\lambda_2S^3_{tot})(S^3_{tot}-(\lambda_3+\mu_3)\rho_3)-d (v_2+w_2)\mu_2S^3_{tot}}{d u_1(v_2+w_2)\mu_2(\lambda_3+\mu_3)\rho_3}.
\end{align*}
It depends both on the reaction rates and the total amounts and may be calculated explicitly. We find that if $\varepsilon^*<0$ (resp.~$\varepsilon^*>1$), then $c_{1,\varepsilon}>c_{2,\varepsilon}$  (resp. $c_{1,\varepsilon}<c_{2,\varepsilon}$) for all $0<\varepsilon<1$, whereas if $0< \varepsilon^*< 1$, the normalized stimulus-response curves intersect for $\varepsilon=\varepsilon^*$, and $c_{1,\varepsilon}>c_{2,\varepsilon}$ (resp. $c_{1,\varepsilon}<c_{2,\varepsilon}$) for $\varepsilon<\varepsilon^*$ (resp.  $\varepsilon>\varepsilon^*$).

We now continue our general investigation of how steady-state concentrations respond to changes in stimulus. For $0<\varepsilon<\delta<1$, we consider the \emph{response coefficient}
\begin{align}\label{Eq:ResponseCoefficient}
\chi_{m,\varepsilon,\delta}=\frac{c_{m,\varepsilon}}{c_{m,\delta}}\qquad\text{for}\qquad 1\le m\le M,
\end{align}
which relates the amount of stimulus required to obtain $\varepsilon$ (resp.~$\delta$) times the maximal possible response in the $m$th layer. Since the response in any layer is an increasing function of the stimulus, it follows that
\begin{align*}
0<\chi_m<1\qquad\text{for}\qquad 1\le m\le M.
\end{align*}

Often one puts $\varepsilon=0.1$ and $\delta=0.9$, and larger values of $\chi_{m,0.1,0.9}$ then indicate a switch-like response in the $m$th layer. In the literature, systems with $\chi_{m,0.1,0.9}>1/81$ are often referred to as \emph{ultrasensitive} \cite{Goldbeter:1981tq}. This is illustrated in a five-layer example in Fig.~\ref{Fig:StimulusResponseM5}, where the intermediate layers (in particular the third) exhibit switch-like behaviors. On the other hand, the top layer shows an almost linear increase in response before reaching a plateau.

\begin{figure}[htbp]
\begin{center}
\includegraphics[width=3.25in]{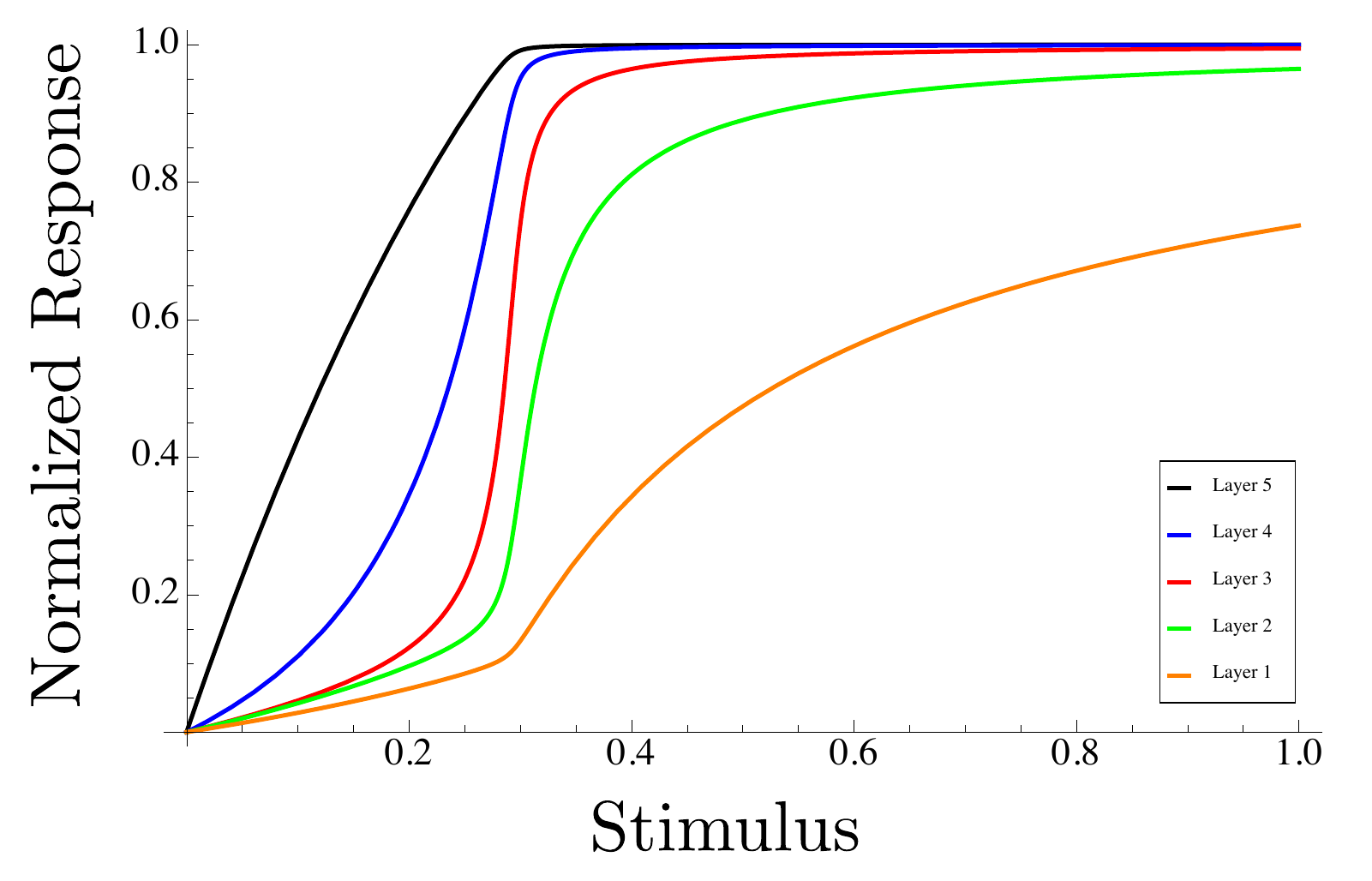}
\end{center}
\caption{
Normalized stimulus-response curves for a five-layer phosphorelay with one phosphorylation site at each layer. Here $S^1_{tot}=S^2_{tot}=S^3_{tot}=S^4_{tot}=10$, $S^5_{tot}=5$, and all reaction rates are set to one. The response in the top layer increases almost linearly before it reaches a plateau. The third layer shows a switch-like (ultrasensitive) behavior.
}\label{Fig:StimulusResponseM5}
\end{figure}

\begin{prop}\label{Prop:SensitivityLastLayer}
For a general $M$-layer phosphorelay we have $\chi_{M,\varepsilon,\delta}<\varepsilon/\delta$.
\end{prop}

\begin{proof}
It follows immediately by using \eqref{Eq:ThePsi0Function} that
\begin{align*}
\chi_{M,\varepsilon,\delta}=\frac{c_{M,\varepsilon}}{c_{M,\delta}}=\frac{\varepsilon}{\delta}\cdot\frac{S^1_{tot}-\lambda_1\delta\rho_M-\mu_1\psi_1(\delta\rho_M)}{S^1_{tot}-\lambda_1\varepsilon\rho_M-\mu_1\psi_1(\varepsilon\rho_M)},
\end{align*}
and since $\psi_1(y)$ according to Proposition \ref{NewProp:PsimSS}  is an increasing function, it follows that in the second fraction, the numerator is smaller than the denominator, and this proves the claim.
\end{proof}

The result in Proposition \ref{Prop:SensitivityLastLayer} shows that for any set of reaction rates and total amounts, the degree of ultrasensitivity in the bottom layer is bounded by the same constant ($\chi_{m,0.1,0.9}<1/9$), and this is thus an intrinsic feature of the phosphorelay.

We are unaware whether the response coefficients $\chi_{m,\varepsilon,\delta}$ are bounded for general $m$. Numerical experiments indicate that also $\chi_{1,\varepsilon,\delta}<\varepsilon/\delta$, but we have not been able to determine this analytically. However, it is possible to calculate the response coefficients in some limit cases, for example when the total amount $S^1_{tot}$ in the top layer is increased or decreased, which could e.g.~be used in an experimental setup where $S^1_{tot}$ can be controlled. The proof of Proposition \ref{Prop:ResponseLimits} below is given in \ref{Appendix:Proofs}.

\begin{prop}\label{Prop:ResponseLimits}
Let $1\le m\le M$ and $1<\varepsilon<\delta<1$, and let all reaction constants and all total amounts except $S^1_{tot}$ be fixed. Then
\begin{align*}
\chi_{m,\varepsilon,\delta}&\rightarrow\frac{\varepsilon(1-\delta)}{\delta(1-\varepsilon)}
\qquad\text{for}\qquad S^1_{tot}\rightarrow 0\\
\chi_{m,\varepsilon,\delta}&\rightarrow
\begin{cases}
\frac{1-\delta}{1-\varepsilon}&\text{if }m=1\\
\frac{\psi_m^{-1}(\varepsilon\psi_m(\xi_1))}{\psi_m^{-1}(\delta\psi_m(\xi_1))}&\text{if }1<m<M\\
\frac{\varepsilon}{\delta}&\text{if }m=M
\end{cases}
\qquad\text{for}\qquad S^1_{tot}\rightarrow\infty,
\end{align*}
and $(\psi_m^{-1}(\varepsilon\psi_m(\xi_1)))/(\psi_m^{-1}(\delta\psi_m(\xi_1)))>\varepsilon/\delta$.
\end{prop}

For $\varepsilon=0.1$ and $\delta=0.9$, the limit $[\varepsilon(1-\delta)]/[\delta(1-\varepsilon)]$ is $1/81$, the common threshold for ultrasensitivity, and for sufficiently high $S^1_{tot}$, the intermediate layers will exhibit higher degrees of ultrasensitivity than the bottom layer. Note that the limits are not necessarily global bounds on the response coefficients as illustrated in Fig.~\ref{Fig:ReponseCoefficientDip}.

\begin{figure}[htbp]
\begin{center}
\includegraphics[width=5.25in]{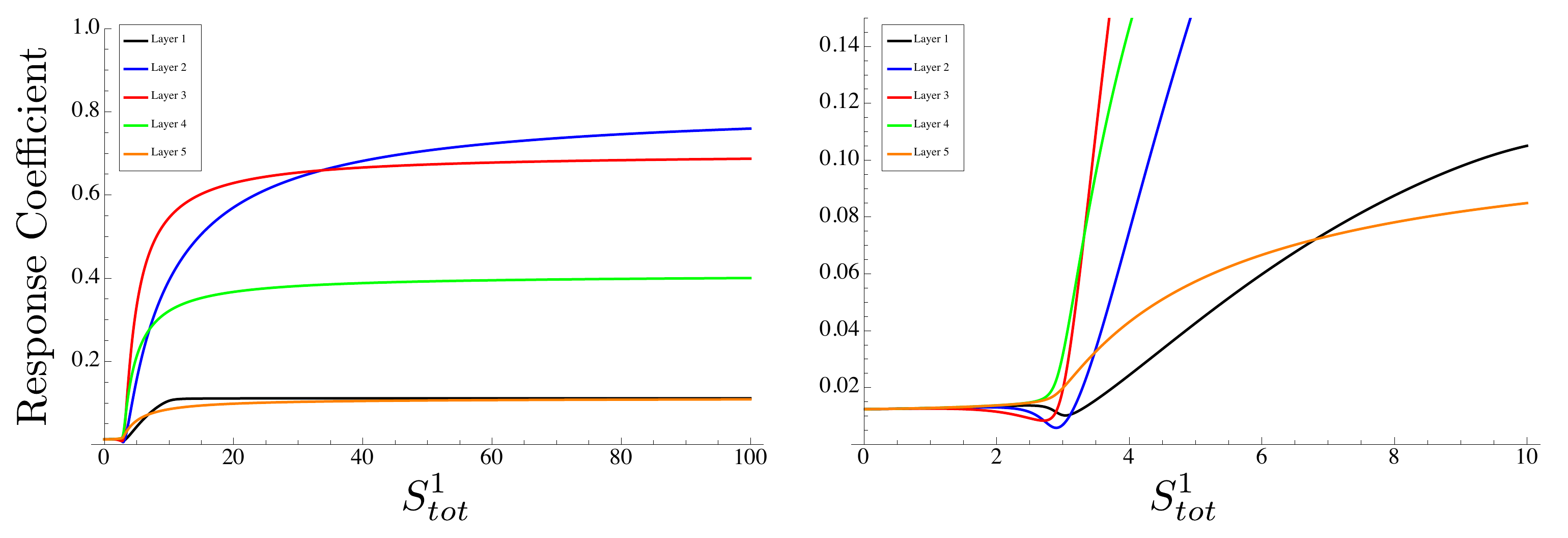}
\end{center}
\caption{
Response coefficients for a five-layer phosphorelay with one phosphorylation site at each layer. Here $S^2_{tot}=S^3_{tot}=S^4_{tot}=10$, $S^5_{tot}=5$, and all reaction rates are set to one. The second plot emphasizes the behavior for small values of $S^1_{tot}$ and reveals that response coefficients are not necessarily increasing functions of $S^1_{tot}$.
}\label{Fig:ReponseCoefficientDip}
\end{figure}

\section{Discussion}

In this paper we have introduced and analytically analyzed a general model of phosphorelays, which extends existing models of phosphorelays \cite{Kim:2006hy,CsikaszNagy:2011dm}, and we have proved the existence and uniqueness of a steady-state. Furthermore, we have derived explicit formulas for the responses in all layers as functions of the stimulus and used these to investigate various aspects of the stimulus-response behavior.

We have showed that the response coefficient in the bottom layer of any phosphorelay is bounded by constants independent of size and architecture of the phosphorelay. Furthermore, we have also demonstrated how qualitatively very different stimulus-response behaviors are possible in layers above the last, and in the three-layer case we have derived an exact condition which distinguishes the three possible scenarios in that case. The variety of behaviors contrasts what has been reported in previous studies using numerical simulations \cite{CsikaszNagy:2011dm}.

The finding that the response coefficient in the bottom layer is bounded is consistent with experimental findings. For example, in the four-layer phosphorelay involved in sporulation initiation of \emph{B. subtilis} (see Figure~\ref{Fig:FourRelays}), it has been observed that the response (the concentration of Spo0A$\sim$P) is only gradually increasing with the stimulus \cite{Fujita:2005fv}. In fact, the authors argue that this is an essential feature of the phosphorelay, since by increasing the stimulus, sporulation is observed, whereas by artificially activating only Spo0A$\sim$P, hardly any sporulation takes place. It has been shown the expression of at least $121$ genes is directly regulated by Spo0A$\sim$P, and that these are activated/repressed at very different concentrations of Spo0A$\sim$P \cite{Molle:2003gk}. It is therefore speculated that the intricate process of sporulation initiation requires several steps mediated by the activity of some of these genes. This suggests the importance of gradually increasing the concentration of Spo0A$\sim$P, since a rapid increase would bypass the intermediate steps \cite{Fujita:2005fv}.

Ultrasensitivity in intermediate layers has previously been suggested using a simpler model \cite{CsikaszNagy:2011dm}, but to our knowledge no experimental studies have determined whether or not this happens \emph{in vitro}. The presence of intermediate layers allows for additional control of the response, where e.g.~an increase in stimulus my be counteracted upon by removal of phosphate in an intermediate layer. For example, there are phosphatases RapA, RapB, and RapE, which are known to dephosphorylate Spo0F$\sim$P, and Spo0E which is known to dephosphorylate Spo0A$\sim$P \cite{Perego:1996wv,Saito:2001vr,Ohlsen:1994up}. We speculate that the ultrasensitivity in intermediate layers is essential in facilitating switch-like cross-talk with external pathways.

We have argued that as the number of layers is increased, the maximal final response decreases. In the example in Figure~\ref{Fig:StripLayers}, the effect appears to saturate already for five layers, a feature not specific to the selected values of reaction constants and total amounts. In fact, the saturation is often observed even earlier. This suggest that the saturation is an intrinsic feature of the phosphorelay structure itself and not the specific reaction rates and total amounts. This fits with the fact that all known phosphorelays to this day contain at most four sites \cite{Appleby:1996wm}.

This paper demonstrates that even relatively complicated systems such as phosphorelays may be treated analytically. Important features of chemical reaction networks may be overlooked if one resorts to numerical simulations alone. Using our approach, previously developed and applied to signaling cascades and enzymatic reactions \cite{Feliu:2011iu,Feliu:2011extra,Feliu:2011wc}, we are able to derive exact and qualitative results about steady-states and stimulus-response behavior for any phosphorelay independent of the number and length of layers, reaction constants, and total amounts of substrate. Also this approach, by providing simple (recursive) expressions relating species concentrations at steady-state, allows for fast and efficient numerical analysis thereby avoiding computationally demanding and error-prone calculations of e.g.~steady-state values.

\appendix

\section{Proofs}
\label{Appendix:Proofs}

\begin{prop}\label{Prop:AppendixLayerEfftects}
If $\lambda_{m_0}$ or $\mu_{m_0}$ is increased, or if the total amount of substrate $S^{m_0}_{tot}$ is decreased, then the maximal response decreases in layers $m\ge m_0$ and increases in layers $m<m_0$.
\end{prop}

\begin{proof}
We add bars over functions and constants in the modified system to distinguish them from the original ones. The function $\psi_m$ only depends on reaction constants and total amounts in layers $m+1,m+2,\ldots,M$. This implies that $\overline{\psi}_m$ corresponding to the modified system is equal to $\psi_m$ for $m\ge m_0$. However, according to \eqref{NewEq:DefinitionOfPsi} either one of the two modifications will decrease the denominator of $\psi_{m_0-1}(y)$ and hence increase $\psi_{m_0-1}(y)$. Therefore $\overline{\psi}_{m_0-1}(y)>\psi_{m_0-1}(y)$ for all $y>0$ in the overlap of the domains of definition of $\overline{\psi}_{m_0-1}$ and $\psi_{m_0-1}$, and using the recursive definition of $\psi_{m-1}$ \eqref{NewEq:DefinitionOfPsi}, we obtain $\overline{\psi}_m(y)>\psi_m(y)$ for all $m<m_0$. Therefore,
\begin{align}\label{Eq:OneIneq}
\overline{S}^{m}_{tot}-\overline{\lambda}_{m}y-\overline{\mu}_{m}\overline{\psi}_{m}(y)
<S^{m}_{tot}-{\lambda}_{m}y-{\mu}_{m}\psi_{m}(y)\quad\text{for}\quad m\le m_0,
\end{align}
where the case $m=m_0$ follows from the assumption that either $\overline{S}^{m_0}_{tot}<S^{m_0}_{tot}$, $\overline{\mu}_{m_0}>\mu_{m_0}$, or $\overline{\lambda}_{m_0}>\lambda_{m_0}$.

Since, by definition, $\xi_{m-1}$ (resp.~$\overline{\xi}_{m-1}$) is the smallest, positive zero of the right-hand side (resp.~left-hand side) of \eqref{Eq:OneIneq}, we see that $\overline{\xi_{m}}<\xi_m$ for all $m<m_0$. In particular, $\overline{\rho}_M=\overline{\xi_0}<\xi_0=\rho_M$, and since the $\overline{\psi}_m$ are unchanged for $m\ge m_0$, we get $\overline{\rho}_m=\overline{\psi}_m(\overline{\rho}_M)={\psi}_m(\overline{\rho}_M)<\psi_m(\rho_M)=\rho_m$ for $m\ge m_0$.

It remains to show that the response decreases in layers upstream from $m_0$, and we first consider the case $m=1$. By the definitions of $\rho_M$ and $\overline{\rho}_M$, we have $S^1_{tot}-\lambda_1\rho_M-\mu_1 \psi_1(\rho_M)=0$ and $S^1_{tot}-\lambda_1\overline{\rho}_M-\mu_1\overline{\psi}_1(\overline{\rho}_M)=0$, so
\begin{align*}
\overline{\rho}_1=\overline{\psi}_1(\overline{\rho}_M)=\frac{S^1_{tot}-\lambda_1\overline{\rho}_M}{\mu_1}
>\frac{S^1_{tot}-\lambda_1\rho_M}{\mu_1}=\psi_1(\rho_M)=\rho_1.
\end{align*}
Using \eqref{NewEq:DefinitionOfPsi} recursively, we see that
\begin{align*}
\psi_{m+1}(\rho_M)=\frac{S^{m+1}_{tot}-\lambda_{m+1}\rho_M}{\mu_{m+1}}-\frac{d\Big(\frac{v_m}{w_n}+1\Big)\rho_M}{u_m\mu_{m+1}\psi_m(\rho_M)},
\end{align*}
which for $m<m_0-1$ does not involve the modified parameters, and since $\psi_1(\overline{\rho}_M)>\psi_1(\rho_M)$, induction shows that $\overline{\rho}_m=\psi_m(\overline{\rho}_M)>\psi_m(\rho_M)=\rho_m$ for all $m<m_0$, which concludes the proof.
\end{proof}

\begin{prop}\label{Prop:AppendixActualLayerEfftects}
If $\lambda_{m_0}$ or $\mu_{m_0}$ is increased, or if the total amount of substrate $S^{m_0}_{tot}$ is decreased, then the response $S^m_{N_m}$ decreases in layers $m\ge m_0$ and increases in layers $m<m_0$.
\end{prop}

\begin{proof}
As already argued in the proof of \ref{Prop:AppendixActualLayerEfftects}, we have $\overline{\psi}_m=\psi_m$ (resp.~$\overline{\psi_m}<\psi_m$) for $m\ge m_0$ (resp.~$m<m_0$) in the overlaps of the domains of definition of the respective functions.

Let $S^m_{N_m}$ and $\overline{S}^m_{N_m}$ be the steady-state values of the response in the $m$th layer in each of the systems. We have $\overline{\psi}_0(\overline{S}^M_{N_M})=c=\psi_0(S^M_{N_M})>\overline{\psi}_0(S^M_{N_M})$, and since $\overline{\psi}_0$ is increasing it follows that $S^M_{N_M}<\overline{S}^M_{N_M}$. Then using that $\psi_m$ is increasing, we get $\overline{S}^m_{N_m}=\overline{\psi}_m(\overline{S}^M_{N_M})=\psi_m(\overline{S}^M_{N_M})>\psi_m(S^M_{N_M})=S^m_{N_m}$ for all $m\ge m_0$. By isolating $\psi_1(y)$ in \eqref{Eq:ThePsi0Function} and using $\overline{\psi}_0(\overline{S}^M_{N_M})=\psi_0(S^M_{N_M})$ and $S^M_{N_M}<\overline{S}^M_{N_M}$, we get
\begin{align*}
\psi_1(S^M_{N_M})&=\frac{1}{\mu_1}\Big(S^1_{tot}-\lambda_1S^M_{N_M}-\frac{dS^M_{N_M}}{\psi_0(S^M_{N_M})}\Big)\\
&<\frac{1}{\mu_1}\Big(S^1_{tot}-\lambda_1\overline{S}^M_{N_M}-\frac{d\overline{S}^M_{N_M}}{\overline{\psi}_0(\overline{S}^M_{N_M})}\Big)=\overline{\psi}_1(\overline{S}^M_{N_M}),
\end{align*}
and from here it follows inductively using \eqref{NewEq:DefinitionOfPsi} as in the proof of \ref{Prop:AppendixLayerEfftects} that $\overline{S}^m_{N_m}=\overline{\psi}_m(\overline{S}^M_{N_M})<\psi_m(S^M_{N_M})=S^m_{N_m}$ for all $m<m_0$.
\end{proof}

\begin{prop}\label{Prop:AppendixResponseLastLayer}
$c_{M,\varepsilon}<c_{m,\varepsilon}$ for all $m< M$ and $0<\varepsilon< 1$.
\end{prop}

\begin{proof}
Using the definitions of $c_{M,\varepsilon}$ and $c_{m,\varepsilon}$, we obtain $c_{M,\varepsilon}=\psi_0(\varepsilon\rho_M)$ and $c_{m,\varepsilon}=\psi_0(\psi_m^{-1}(\varepsilon \psi_m(\rho_M)))$, and since both $\psi_m$ and $\psi_0^{-1}$ are increasing functions, $c_{M,\varepsilon}<c_{m,\varepsilon}$ if and only if $\psi_m(\varepsilon\rho_M)<\varepsilon \psi_m(\rho_M)$.

In the case $m=M$, we have $\psi_M(\varepsilon\rho_M)=\varepsilon\rho_M=\varepsilon\psi_M(\rho_M)$. Hence it suffices to prove that $\psi_m(\varepsilon\rho_M)\le \varepsilon\psi_m(\rho_M)$ implies $\psi_{m-1}(\varepsilon\rho_M)<\varepsilon\psi_{m-1}(\rho_M)$, and this follows using $\varepsilon\rho_M<\rho_M$, since
\begin{align*}
\psi_{m-1}(\varepsilon\rho_M)&=\frac{d\big(\frac{v_{m-1}}{w_{m-1}}+1\big)\varepsilon\rho_M}{u_{m-1}\big(S^m_{tot}-\lambda_m \varepsilon\rho_M-\mu_m \psi_m(\varepsilon\rho_M)\big)}\\
&<\varepsilon\frac{d\big(\frac{v_{m-1}}{w_{m-1}}+1\big)\rho_M}{u_{m-1}\big(S^m_{tot}-\lambda_m \rho_M-\mu_m \psi_m(\rho_M)\big)}=\varepsilon\psi_{m-1}(\rho_M),
\end{align*}
which finishes the proof.
\end{proof}

\begin{lemma}\label{Lemma:RhoLimit}
Let all reaction rates and all total amounts except $S^1_{tot}$ be fixed. Then $\rho_M\rightarrow\xi_1$  for $S^1_{tot}\rightarrow\infty$ and $\rho_M\rightarrow 0$ for $S^1_{tot}\rightarrow 0$.
\end{lemma}

\begin{proof}
By definition, $\rho_M$ is a root in the denominator of \eqref{Eq:ThePsi0Function}, and hence $S^1_{tot}=\lambda_1\rho_M+\mu_1\psi_1(\rho_M)$. This function is continuous and increasing on $[0,\xi_1)$, equals $0$ for $\rho_M=0$, and tends to infinity when $\rho_M\rightarrow\xi_1^-$.  The function is thus invertible, say $\rho_M=\phi(S^1_{tot})$, such that $\phi$ is continuous and increasing with $\phi(0)=0$ and satisfies $\phi(S^1_{tot})\rightarrow\xi_1$ for $S^1_{tot}\rightarrow\infty$.
\end{proof}

\begin{lemma}\label{Lemma:DifferentialLemma}
For all $m<M$, the functions $\psi_{m}$ and $\psi_{m}^{-1}$ satisfy
\begin{align*}
\psi_{m}'(0)=d\Big(\frac{v_m}{w_m}+1\Big)\Big/u_mS^{m+1}_{tot}\quad\text{and}\quad (\psi_{m}^{-1})'(0)=u_mS^{m+1}_{tot}\Big/d\Big(\frac{v_m}{w_m}+1\Big).
\end{align*}
\end{lemma}

\begin{proof}
The first statement follows by differentiating the recursive expression for $\psi_m(y)$ \eqref{NewEq:DefinitionOfPsi} and inserting $y=0$, and then the second statement immediately follows from $(\psi_{m}^{-1})'(0)=1/\psi_{m}'(\psi_{m}^{-1}(0))=1/\psi_{m}'(0)$.
\end{proof}

\begin{lemma}\label{Lemma:SecondDifferentialLemma}
As a function of $S^1_{tot}$, the maximal response $\rho_M$ satisfies
\begin{align}
\frac{d\rho_M}{dS^1_{tot}}(0)=\frac{u_1S^2_{tot}}{\lambda_1 u_1S^2_{tot} + \mu_1 d \Big(\frac{v_1}{w_1}+1\Big)}
\end{align}
\end{lemma}

\begin{proof}
By differentiating $S^1_{tot}=\lambda_1\rho_M+\mu_1\psi_1(\rho_M)$ with respect to $S^1_{tot}$, we get $\rho_M'(0)=(1-\mu_1\psi_1'(\rho_M(0))\rho_M'(0))/\lambda_1=(1-\mu_1\psi_1'(0)\rho_M'(0))/\lambda_1$, and by isolating $\rho_M'(0)$ and inserting \ref{Lemma:DifferentialLemma}, we obtain the desired result.
\end{proof}

By plugging the expression \eqref{Eq:AllMaximalResponses} into the definition of the response coefficient \eqref{Eq:ResponseCoefficient}, it follows that we have
\begin{align*}
\chi_{m,\varepsilon,\delta}=\underbrace{\frac{\psi_m^{-1}(\varepsilon\psi_m(\rho_M))}{\psi_m^{-1}(\delta\psi_m(\rho_M))}}_{\alpha_m}
\cdot
\underbrace{\frac{S^1_{tot}-\lambda_1\psi_m^{-1}(\delta\psi_m(\rho_M))-\mu_1\psi_1(\psi_m^{-1}(\delta\psi_m(\rho_M)))}{S^1_{tot}-\lambda_1\psi_m^{-1}(\varepsilon\psi_m(\rho_M))-\mu_1\psi_1(\psi_m^{-1}(\varepsilon\psi_m(\rho_M)))}}_{\beta_m},
\end{align*}
and below we will consider the factors $\alpha_m$ and $\beta_m$ separately.

\begin{prop}\label{Prop:AlphaProp}
Let all reaction rates and all total amounts except $S^1_{tot}$ be fixed. Then $\alpha_M=\varepsilon/\delta$, and
\begin{align*}
\alpha_m&\rightarrow
\frac{\varepsilon}{\delta}
\qquad\text{for}\qquad
S^1_{tot}\rightarrow 0.\\
\alpha_m&\rightarrow
\begin{cases}
1&\text{if }m=1\\
\frac{\psi_m^{-1}(\varepsilon\psi_m(\xi_1))}{\psi_m^{-1}(\delta\psi_m(\xi_1))}&\text{if }1<m<M
\end{cases}
\qquad\text{for}\qquad
S^1_{tot}\rightarrow\infty,
\end{align*}
and $(\psi_m^{-1}(\varepsilon\psi_m(\xi_1)))/(\psi_m^{-1}(\delta\psi_m(\xi_1)))>\varepsilon/\delta$.
\end{prop}

\begin{proof}
Since $\psi_M=\operatorname{id}$, it immediately follows that $\alpha_M=(\varepsilon\rho_M)/(\delta\rho_M)=\varepsilon/\delta$ as claimed. According to \ref{Lemma:RhoLimit}, we have $\rho\rightarrow 0$ for $S^1_{tot}\rightarrow 0$, so since $\psi_m(0)=0$, it follows that $\alpha_m$ is a $0/0$--expression in the limit $S^1_{tot}\rightarrow 0$, and we may apply L'H\^opital's rule. For the numerator of $\alpha_m$ (and analogously for the denominator) we have
\begin{align*}
\frac{d}{dS^1_{tot}}\Big(\psi_m^{-1}(\varepsilon \rho_m)\Big)
=\frac{d\psi_m^{-1}}{dx}(\varepsilon \rho_m)\varepsilon\frac{d\rho_m}{dS^1_{tot}},
\end{align*}
and using that the results of \ref{Lemma:DifferentialLemma} and \ref{Lemma:SecondDifferentialLemma} are non-zero, it follows that
\begin{align*}
\frac{d}{dS^1_{tot}}\Big(\psi_m^{-1}(\varepsilon \rho_m)\Big)\Big/\frac{d}{dS^1_{tot}}\Big(\psi_m^{-1}(\delta \rho_m)\Big)
\rightarrow\frac{\varepsilon}{\delta}
\qquad\text{for}\qquad S^1_{tot}\rightarrow 0.
\end{align*}

For the limit $S^1_{tot}\rightarrow\infty$, we consider first the case $m=1$. According to \ref{Lemma:RhoLimit} we have $\rho_1=\psi_1(\rho_M)\rightarrow\infty$ for $S^1_{tot}\rightarrow\infty$, and hence $\alpha_1=\psi_1^{-1}(\varepsilon\rho_1)/\psi_1^{-1}(\delta\rho_1)\rightarrow\xi_1/\xi_1=1$,  The remaining cases follow from the first part of \ref{Lemma:RhoLimit}. Finally, since $\psi_m$ is convex and increasing, the inverse $\psi_m^{-1}$ is concave and increasing. Using $\psi_m^{-1}(0)=0$, it follows that $\psi_m^{-1}(x)/x>\psi_m^{-1}(y)/y$ for all $x<y$. In particular, this holds for $x=\varepsilon\rho_m$ and $y=\delta\rho_m$, and hence $(\psi_m^{-1}(\varepsilon\psi_m(\xi_1)))/(\psi_m^{-1}(\delta\psi_m(\xi_1)))>\varepsilon/\delta$.
\end{proof}

\begin{prop}\label{Prop:BetaProp}
Let all reaction rates and all total amounts except $S^1_{tot}$ be fixed. Then
\begin{align*}
\beta_m&\rightarrow
\frac{1-\delta}{1-\varepsilon}
\qquad\text{for}\qquad
S^1_{tot}\rightarrow 0\\
\beta_m&\rightarrow
\begin{cases}
\frac{1-\delta}{1-\varepsilon}&\text{if }m=1\\
1&\text{if }1<m\le M
\end{cases}
\qquad\text{for}\qquad
S^1_{tot}\rightarrow\infty.
\end{align*}
\end{prop}

\begin{proof}
To simplify notation, we denote by $f_{m,\delta}$ and $f_{m,\varepsilon}$ the numerator and denominator, respectively, of $\beta_m$.
Note that $\beta_m$ is a $0/0$--expression in the limit $S^1_{tot}\rightarrow 0$, and hence we may use L'H\^opital's rule. Let $z_{m,\delta}=\psi_m^{-1}(\delta \rho_m)$ such that
$f_{m,\delta}=S^1_{tot}-\lambda_1 z_{m,\delta} - \mu_1 \psi_1(z_{m,\delta})$. Then
\begin{align*}
\frac{df_{m,\delta}}{dS^1_{tot}}=1-\lambda_1\frac{dz_{m,\delta}}{dS^1_{tot}}-\mu_1\frac{d\psi_1}{dy}(z_m)\frac{dz_{m,\delta}}{dS^1_{tot}}
=1-\frac{dz_{m,\delta}}{dS^1_{tot}}\Big(\lambda_1+\mu_1\frac{d\psi_1}{dy}(z_{m,\delta})\Big).
\end{align*}
Now note that
\begin{align*}
\frac{dz_{m,\delta}}{dS^1_{tot}}=\delta \frac{d\psi_m^{-1}}{dx}(\delta\psi_m(\rho_M)) \frac{d\psi_m}{dy}(\rho_M)\frac{d\rho_M}{dS^1_{tot}},
\end{align*}
and by plugging in the results from \ref{Lemma:DifferentialLemma} and \ref{Lemma:SecondDifferentialLemma}, it follows that $df_{m,\delta}/dS^1_{tot}\rightarrow 1-\delta$ for $S^1_{tot}\rightarrow 0$. Similarly, $df_{m,\varepsilon}/dS^1_{tot}\rightarrow 1-\varepsilon$, and hence L'H\^opital's rule implies that $\beta_m\rightarrow (1-\delta)/(1-\varepsilon)$ for $S^1_{tot}\rightarrow 0$.

We have $f_{1,\delta}=S^1_{tot}-\lambda_1\psi_1^{-1}(\delta \rho_1)-\mu_1\delta \rho_1$, and by combining this with $S^1_{tot}=\lambda_1\rho_M+\mu_1\rho_1$, it follows that $f_{1,\delta}=\lambda_1(\rho_M-\psi_1^{-1}(\delta \rho_1))+\mu_1(1-\delta)\rho_1$ and hence
\begin{align*}
\beta_1=
\frac{\lambda_1(\rho_M-\psi_1^{-1}(\delta \rho_1))+\mu_1(1-\delta)\rho_1}{\lambda_1(\rho_M-\psi_1^{-1}(\varepsilon \rho_1))+\mu_1(1-\varepsilon)\rho_1}
\rightarrow \frac{1-\delta}{1-\varepsilon}\qquad\text{for}\qquad S^1_{tot}\rightarrow\infty,
\end{align*}
since $\rho_M$ and $\psi_1^{-1}$ are bounded and $\rho_1\rightarrow\infty$ for $S^1_{tot}\rightarrow\infty$. For $m>1$, we have $\psi_m(\rho_M)\rightarrow\psi_m(\xi_1)$ for $S^1_{tot}\rightarrow\infty$, and it immediately follows that
\begin{align*}
\beta_m=\frac{S^1_{tot}-\lambda_1\psi_m^{-1}(\delta\psi_m(\rho_M))-\mu_1\psi_1(\psi_m^{-1}(\delta\psi_m(\rho_M)))}{S^1_{tot}-\lambda_1\psi_m^{-1}(\varepsilon\psi_m(\rho_M))-\mu_1\psi_1(\psi_m^{-1}(\varepsilon\psi_m(\rho_M)))}\rightarrow 1\quad\text{for}\quad S^1_{tot}\rightarrow\infty,
\end{align*}
since the last two terms in both numerator and denominator are bounded.
\end{proof}

By combining \ref{Prop:AlphaProp} and \ref{Prop:BetaProp}, we obtain Proposition \ref{Prop:ResponseLimits} from the main text.

\section{Example}
\label{Appendix:Example}

Here we provide details about the calculations left out in the example in Section \ref{Sec:StimulusResponse}. Using the definitions of $\psi_1$ and $\psi_2$, we may write $\Delta(\varepsilon)$ on the form
\begin{align*}
\Delta(\varepsilon)=\frac{r_3\varepsilon^3+r_2\varepsilon^2+r_1\varepsilon}{t_2\varepsilon^2+t_1\varepsilon+t_0}=\frac{R(\varepsilon)}{T(\varepsilon)},
\end{align*}
for example using computer software capable of symbolic manipulation, where the coefficients $r_i$ and $t_i$ only depend on the reaction constants, the total amounts, and the maximal response $\rho_3$. The coefficients of the polynomials may be chosen such that $t_2<0$ and both $T(0)$ and $T(1)$ are positive. Hence $T$ is a second degree polynomial with negative leading coefficient, and it assumes positive values in the end points $\varepsilon=0$ and $\varepsilon=1$. Therefore $T(\varepsilon)>0$ for all $0<\varepsilon<1$, and hence $\sign\Delta(\varepsilon)=\sign R(\varepsilon)$.

Since $\Delta(0)=\Delta(1)=0$, it follows that $R(0)=R(1)=0$, and by factoring these trivial roots we obtain the last root as stated in the main text. For some constant $\tau<0$, the polynomial $R(\varepsilon)$ factors $R(\varepsilon)=\tau\cdot\varepsilon(\varepsilon-1)(\varepsilon-\varepsilon^*)=\tau\varepsilon^3-\tau(1+\varepsilon^*)\varepsilon^2+\tau\varepsilon^*\varepsilon$, and hence $\sign R'(0) = -\sign(\varepsilon^*)$. Summing up, we know all three roots of $R(\varepsilon)$ as well as the slope of $R(\varepsilon)$ at $\varepsilon=0$, which completely determines the sign of $R(\varepsilon)$ at any point.

\section{Convergence and stability of the steady-state}
\label{Appendix:Stability}

We prove the convergence to a unique BMSS of the phosphorelay for any set of positive initial conditions using Theorem 2 in \cite{Angeli:2010ff}. We will restate the theorem here, but first we introduce some concepts from \cite{Angeli:2010ff}. To simplify notation, we use $\leftrightarrow$ to denote reversible reactions in inline text.

For a reaction $A\rightarrow B$, $A$ is the \emph{reactant} and $B$ is the \emph{product}, and for every reversible reaction, e.g.~$S_n^{m} \leftrightarrow S_{n+1}^m$, a direction is chosen so that the reactant and the product are well-defined. For the reversible reactions in \eqref{Reactions:IntraSubstrate}, \eqref{Reactions:InterSubstrate}, and \eqref{Reactions:Special}, we choose the left-hand side to be the reactant and the right-hand side to be the product.
We have a total of $n_S=\sum_{m=1}^M N_m+M-1$ species and $n_R=\sum_{m=1}^M N_m+M$ reactions in the system. Define the $n_S\times n_R$ \emph{stoichiometric matrix} $\Gamma$ 
such that the entry $\Gamma_{s,r}$ corresponding to species $s$ and reaction $r$ is $1$ if $s$ is in the reactant of $r$, $-1$ if $s$ is in product, and zero otherwise. Here orders on the sets of species and reactions are implicitly chosen. 

The \emph{directed SR-graph} is constructed as follows: The set of vertices is the union of the set of species (called \emph{species nodes}) and the set of reactions (called \emph{reaction nodes}). If a species $s$ takes part in a reversible reaction $r$ or is part of the reactant of an irreversible reaction $r$, there are edges $s\rightarrow r$ and $r\rightarrow s$. If $s$ is part of the product of an irreversible reaction $r$, there is an edge $r\rightarrow s$. A \emph{siphon} $\Sigma$ is a non-empty subset of species such that if $s\in\Sigma$ is in the product of a reaction $r$, then $\Sigma$ contains at least one species in the reactant of $r$. 
Here reversible reactions are considered as two different irreversible reactions, so that each side of the reversible reaction appears as product in one reaction and as reactant in the other. A siphon is \emph{minimal} if it contains no siphon other than itself. 

Theorem 2 in \cite{Angeli:2010ff} states that all solutions  of the phosphorelay ODEs  in  $\R^{n_S}_{>0}$ converge to a unique equilibrium if the following four conditions hold:
\begin{enumerate}[(i)]
\item The system of  ODEs of the phosphorelay is persistent. 
\item For all species $s$ and reactions $r_1\neq r_2$, the product $-\Gamma_{s,r_1}\Gamma_{s,r_2}$ is non-negative. 
\item There is a directed path between any two reaction nodes in the directed SR-graph.
\item The kernel of $\Gamma$ contains a positive vector.
\end{enumerate}

\emph{Remark.} With the notions from \cite{Angeli:2010ff}, conditions (ii) and (iii) imply that the corresponding system in reaction coordinates is monotone with respect to the positive orthant cone, and strongly monotone in the interior with respect to that order. This is a consequence of Proposition 5.3 and the proof of Theorem 1 in \cite{Angeli:2010ff}.

We will now prove (i)--(iv). By the choice of directions of reactions in our system, each species is on the left of exactly one reaction and on the right of exactly one reaction. It follows that for each species $s$ there are exactly two reactions $r_1,r_2$ such that $\Gamma_{s,r_1}, \Gamma_{s,r_2}\neq 0$ and further that they have opposite sign. Thus, $\Gamma_{s,r_1}\Gamma_{s,r_2}<0$ and zero for all other choices of reactions. This proves (ii). 
Each row of $\Gamma$ has only two non-zero entries, and they are of opposite signs. Therefore the vector $(1,\dots,1)$ belongs to the kernel of $\Gamma$, and hence (iv) holds.

To show (iii), note that for a fixed $m$, there is a path in the directed SR-graph between any two reaction nodes of the form
\begin{align*}
r_{n,m}\colon S_n^{m}\leftrightarrow S_{n+1}^m,\:\: e_m\colon S_{N_m}^{m}+S_0^{m+1}\leftrightarrow X^m,\:\:\text{or}\:\: d_m\colon X^m\rightarrow S_0^m+S_{1}^{m+1}.
\end{align*}
For $m=M$ the statement is true with the last two reactions replaced by $d_0\colon S_{N_M}^M\rightarrow S_0^M$. Furthermore, there is a path from $w_m$ to $r_{1,m+1}$ and to $e_{m-1}$, connecting reactions in different layers. There is a path from the reaction node $e_0\colon \xymatrix@C=12pt{S_{0}^1 \ar[r] & S_1^1}$ to $r_{1,1}$ and from $d_1$ to $e_0$. Therefore, a directed path between any two reactions of the phosphorelay exists, and (iii) holds.

All that is left is to prove (i). For that, we use Theorem 2 in \cite{Angeli:2007ig} that states that if (a) the network has a positive conservation law, and (b) there is a conservation law with non-negative coefficients on the species for each minimal siphon, then the network is persistent. Since each species of the phosphorelay is part of a conservation law with non-negative coefficients \eqref{Eq:ConservationLaws}, it follows that (a) holds. 

If we show that the sets $\Sigma_m=\{S^m_0,S^m_1,\ldots,S^m_{N_m},X^{m-1},X^m\}$ for $m=1,2,\ldots,M$ (removing the non-defined $X^0,X^M$ for $m=1,M$) are the minimal siphons, then the conservation law \eqref{Eq:ConservationLaws} ensures that (b) holds and the proof is completed. 
We construct a graph that gives an easy visual inspection of which  the minimal siphons of the phosphorelay are. If $r$ is a reaction that contains a species  $s_1$ in the reactant and a species $s_2$ in the product, then we draw an edge $s_1\xrightarrow{r}s_2$  with label $r$. The graph is

$$
\xymatrix@C=7.5pt{
& S_0^1 \ar[dl] & & S_0^{m-1} & & &  S_0^m \ar@<0.3ex>[dll]^{e_{m-1}}  &  & & S_0^{m+1}  \ar@<0.3ex>[dl]^{e_{m}}  & & S_0^M \\ 
S_1^1\ar@{<->}[dr]   && \dots & &  X^{m-1} \ar[dr]^{d_{m-1}} \ar@<0.3ex>[urr]^{e_{m-1}'}  \ar[ul]^{d_{m-1}} \ar@<0.3ex>[dl]^{e_{m-1}'}  &&  & & X^{m} \ar@<0.3ex>[ur]^{e_m'}    \ar[dr]^{d_{m}}  \ar[ull]_{d_{m}} \ar@<0.3ex>[dl]^{e_{m}'} & & \dots & & S_{N_M}^M \ar[ul] \\
& S_2^1& & S_{N_{m-1}}^{m-1} \ar@<0.3ex>[ur]^{e_{m-1}}  & & S_{1}^m \ar@{<->}[r] & \dots \ar@{<->}[r] & S_{N_m}^m \ar@<0.3ex>[ur]^{e_{m}}  & & S_{1}^{m+1} & & S_{N_{M}-1}^{M}\ar@{<->}[ur]
}
 $$ 
where $d_m$ is as above, $e_m\colon S_{N_m}^{m}+S_0^{m+1} \rightarrow X^m$, and $e_m'\colon X^m\rightarrow S_{N_m}^{m}+S_0^{m+1}$. The labels of the reactions $r_{n,m}\colon S_n^{m}\leftrightarrow S_{n+1}^m$ are not shown.

Let $\Sigma$ be a siphon. Then inspection of the graph gives:
\begin{enumerate}[(1)]
\item  If  $S_n^m$ belongs to $\Sigma$  for some $n>0$, then so do $S_{n'}^m$ for all $n'>0$. 
\item If  $S_n^m$ belongs to $\Sigma$ for some $n\geq 0$, then so do $X^m$ and $X^{m-1}$. Further, if $S_n^1$ or $S_0^M$ belong to $\Sigma$ then so do $S_0^1$ and $S_{N_M}^M$, respectively.
\item  If $X^m$ belongs to $\Sigma$, then either $S_0^{m+1}$ or $S_{N_m}^m$ (and thus $S_n^m$ for all $n>0$) belong to $\Sigma$.
\end{enumerate}
It is easy to see that the middle pentagon which contains the species in $\Sigma_m$ is a siphon for all $1\le m\le M$, and conditions (1)--(3) ensure that it is minimal. If  $\Sigma$  is a siphon that does not contain   $\Sigma_m$ for any $m$, then (1)--(3) imply that it must contain $X^m$ for all $m$.
If $S_0^m$ does not belong to $\Sigma$ for any $m$, then by (3) $S_n^m$ belongs to $\Sigma$ for all $m$ and $n>0$. It follows from  (2) that so does $S_0^1$ and $\Sigma_1\subseteq \Sigma$, thereby reaching a contradiction.
Thus, there is an $m$ for which $S_0^m$ belongs to $\Sigma$. Since $S_{N_m}^m$ does not ($\Sigma_m\nsubseteq \Sigma$), by (3) $S_0^{m+1}$ belongs to $\Sigma$. We repeat the argument to conclude that $S_0^M\in \Sigma$. It follows from  (2) that so does $S_{N_M}^M$ and $\Sigma_M\subseteq \Sigma$, again reaching a contradiction. Therefore, any siphon contains $\Sigma_m$ for some $m$ as desired.



\bibliographystyle{elsarticle-num}
\bibliography{RelayBib}







\end{document}